\documentclass[a4paper, 12pt]{article}

\usepackage[sort&compress]{natbib}
\bibpunct{(}{)}{;}{a}{}{,} 

\usepackage{float}
\usepackage{amsthm, amsmath, amssymb, mathrsfs, multirow, url}
\usepackage{graphicx} 
\usepackage{ifthen} 
\usepackage{amsfonts}
\usepackage[usenames]{color}
\usepackage{fullpage}
\usepackage{subfig}
\usepackage{textcomp}
\usepackage{multirow}
\usepackage{dsfont}
\usepackage{tabularx}

\theoremstyle{plain} 
\newtheorem{thm}{Theorem}

\newtheorem{prop}{Proposition}

\theoremstyle{definition}

\theoremstyle{remark} 
\newtheorem*{pralg}{Predictive Recursion Algorithm}

\newcommand{\unif}{{\sf Unif}}

\newcommand{\bet}{{\sf Beta}}

\newcommand{\RR}{\mathbb{R}}

\newcommand{\XX}{\mathbb{X}}
\newcommand{\YY}{\mathbb{Y}}

\newcommand{\FF}{\mathbb{F}}

\renewcommand{\S}{\mathcal{S}}
\renewcommand{\SS}{\mathbb{S}}

\title{Estimating a mixing distribution on the sphere using predictive recursion}
\author{Vaidehi Dixit\footnote{Department of Statistics, North Carolina State University; {\tt vdixit@ncsu.edu}, {\tt rgmarti3@ncsu.edu}} \quad and \quad Ryan Martin$^*$}
\date{\today}

\begin{document}

\maketitle 

\begin{abstract}
Mixture models are commonly used when data show signs of heterogeneity and, often, it is important to estimate the distribution of the latent variable responsible for that heterogeneity.  This is a common problem for data taking values in a Euclidean space, but the work on mixing distribution estimation based on directional data taking values on the unit sphere is limited.  In this paper, we propose using the predictive recursion (PR) algorithm to solve for a mixture on a sphere.  One key feature of PR is its computational efficiency.  Moreover, compared to likelihood-based methods that only support finite mixing distribution estimates, PR is able to estimate a smooth mixing density.  PR's asymptotic consistency in spherical mixture models is established, and simulation results showcase its benefits  compared to existing likelihood-based methods.  We also show two real-data examples to illustrate how PR can be used for goodness-of-fit testing and clustering. 

\smallskip

\emph{Keywords and phrases:} Directional data; EM algorithm; marginal likelihood; mixture model; von Mises--Fisher distribution.
\end{abstract}

\section{Introduction}

Directional data occur in the natural sense like directions of fibres in a material or in the manufactured sense when observed data vectors are normalized. Such kind of data need to be modeled with special kind of distributions which are supported on compact manifolds. Significant contributions to this field appeared from 1950s onward as mentioned in \citet{fisher1993}. Various directional distributions and corresponding statistical methods have been developed.  However, these methods fall short when directional data cannot be adequately modeled by a single distribution on the sphere.  For example, diffusion MRI data may reveal clusters of fascicles of fibers in brain tissue pointing in generally different directions, and \citet{shakya2017} and \citet{yan2018} have used mixtures to accommodate this heterogeneity.  Similarly, \citet{franke2016} have proposed the use of mixture distributions to model directional data from fiber composites with the goal of estimating underlying fiber directions.  These references rely primarily on classical statistical methods for fitting finite mixture models for directional data.  The goal of this paper is to develop methods for nonparametric estimation of a general mixing distribution having a smooth density function supported on an unit sphere.

%have different mean directions, mixture distributions have been used to uncover the true fibre directions using {\color{red} dMRI} data as in \citet{shakya2017}, \citet{yan2018} to name a few. To understand the micro structure of foam materials, \citet{franke2016} propose the use of Schladitz mixtures to model directional data with multiple modes obtained from fibre composites. 

%goodness of fit testing and general inference have been discussed. These methods fall short when directional data cannot be modeled by a single distribution.  

%These methods assume discrete mixtures on the sphere, and estimate parameters of the kernel distribution using the EM algorithm. The main drawback of this method is that it has to be remodeled for every kernel density function which will change based on data. Furthermore, since the mixing density is assumed to be discrete, the number of mixture components are to be estimated by cross validation. We propose using the \textit{Predictive Recursion}(PR) algorithm a nonparametric method to solve a general mixture as in \eqref{eq:mixture}.

To set the scene, suppose we have $Y^n = (Y_1,\ldots,Y_n)$ independent and identically distributed (iid) observations from a mixture $f(y)$,
\begin{equation}
\label{eq:mixture}
f(y) = f_\Psi(y) = \int_\XX k(y \mid x) \, \Psi(dx), \quad y \in \YY, 
\end{equation}
where $k$ is a known kernel, i.e., $y \mapsto k(y \mid x)$ is a density for each $x \in \XX$, and $\Psi$ is an unknown mixing distribution, a probability measure defined on (a sigma-algebra of subsets of) the space $\XX$.  The goal is to estimate $\Psi$.  

For the mixture in \eqref{eq:mixture}, the most common mixture model problem is when the mixing distribution $\Psi$ is {\em discrete}, i.e., where $\Psi$ is assumed to have a finite (or perhaps countable) support. For these finite mixtures, the problem is reduced to estimation of component proportions and corresponding parameters.  In this case, the standard estimation procedure is maximum likelihood, often via the expectation--maximization (EM) algorithm; see, for example, \citet{dempsteretal1977}, \citet{mclachlanpeel2000}, and \citet{teeletal2015}. Bayesian approaches have also been explored where using a prior for $\psi(x)$ and a data-augmentation strategy \citep{vm2001} an EM-like method is used. A natural scenario is introducing a penalty or a prior for the unknown number of mixture components like in \citet{leroux1992} and \citet{richardsongreen1997}.
% {\color{red} Brief summary of Bayesian approaches for finite mixtures (via data augmentation and Gibbs sampling) with some references...}  
% {\color{red} Some details about nonparametric maximum likelihood estimation (e.g., Lindsay's work)

For the general case where $\Psi$ in \eqref{eq:mixture} need not be discrete, maximum likelihood and Bayesian estimation is still possible.   The nonparametric maximum likelihood estimation explained in \citet{lindsay1995} is one such approach. \citet{escobar1995} use a Dirichlet process prior for the mixing distribution in a normal mixture model, resulting in the now-common Dirichlet process mixture of normals model. However, while these procedures do not directly assume $\Psi$ to be discrete, the estimators end up being discrete, or effectively so.  Indeed, \citet[][Theorem~20]{lindsay1995} shows that the nonparametric maximum likelihood estimator of $\Psi$ is almost surely discrete, and the posterior mean of $\Psi$ under the Dirichlet process mixture model tends to be too rough \citep[cf.,][Fig.~1]{tmg} and not intended for estimating a smooth density function.

If one genuinely desires a smooth estimator of the mixing distribution, then the choices are limited.  One option is to suitably ``smooth out'' the nonparametric maximum likelihood estimator, either by manually smoothing \citep[e.g.,][]{eggermont1995} or by introducing a roughness penalty \citep[e.g.,][]{liuetal2009}.  Another option is the {\em predictive recursion} (PR) algorithm first developed by Michael Newton and co-authors \citep[e.g.,][]{newtonzhang1999, newtonetal1998, newton2002}.  This algorithm was motivated by some interesting connections to the Dirichlet process mixtures, and the estimator's key features are its speed, ease of computation, and ability to estimate a mixing density with respect to any user-specified dominating measure.  Theory was developed for the PR estimator later by \citet{ghoshtokdar2006}, \citet{martinghosh2008}, \citet{tmg}, and \citet{martintokdar2009}; extensions and PR-based methods have been proposed in \citet{martintokdar2011}, \citet{martintokdar2012} \citet{martinhan2016}, \citet{tansey2018} and \citet{scott2015}. Despite all of these efforts, to our knowledge, the PR algorithm or, more generally, mixing density estimation has not been investigated in the literature on directional data. This paper aims to address some of the challenges associated with mixtures on a sphere. Estimation via PR involves integrating over the support of the mixing density, which in this case is a sphere. This poses a computational challenge as explored in \ref{SS:prcomp}. We also discuss consistency of the PR estimate on a sphere and verification of the required assumptions.

We organize the rest of the paper as follows. Section~\ref{S:background} gives a background on spherical distributions, their properties and estimation in mixtures. The proposed method of using PR on a sphere is given in Section~\ref{S:pr}. We compare the mixture estimates with those obtained from EM via simulations in Section~\ref{simulations}. Section~\ref{S:extensions} illustrates, using real data, how the PR-based methodology can be used for goodness-of-fit testing and for clustering surface normals (three-dimensional unit vectors) to reconstruct images. We conclude, in Section~\ref{S:discuss}, with a summary and a brief discussion of some open problems.

%For the general problem in \eqref{eq:mixture}, the mixing distribution is not necessarily discrete, so one might desire a nonparametric estimator of $\Psi$ that has a smooth density respect to some specified dominating measure. Classical likelihood techniques have been used in obtaining a smooth estimator for $\Psi$ by smoothing the nonparametric maximum likelihood estimator like in  and \citet{liuetal2009}. 

%The \textit{Predictive recursion} (PR) algorithm is one such method that is designed to generate an estimate of $\Psi$ that has a density with respect to any specified dominating measure.  In other words, 

%This estimation method doesn't need to be modified for a different kernel and takes into account a nonparametric structure of the mixing density. The PR algorithm was proposed in \citet{newtonzhang1999} with discussions following in \citet{newton2002} and \citet{tmg}. 

\section{Background}
\label{S:background}

\subsection{Distributions on the sphere}
\label{SS:sphere.dist}

In most applications, a random vector has a density function with respect to a product measure, e.g., Lebesgue measure, on the sample space.  However, for data taking values on a sphere, $\SS$, there is a constraint on the sample space that makes it impossible to define a density with respect to a product measure.  For our present application, we use the surface area measure, denoted by $\sigma$, which coincides with the Haar measure on the special orthogonal group of rotations on the three-dimensional sphere.    

The simplest way to describe the surface area measure is via a change-of-variables operation.  Any point on the sphere can be expressed in spherical coordinates.  Identify a point $x=(x^{(1)},x^{(2)},x^{(3)})^\top$ on the unit sphere $\SS \subset \RR^3$ with its polar coordinates $(\theta,\phi) \in [0,\pi] \times [0,2\pi)$ according to the ISO convention.  That is, $\theta$ represents the polar angle between $x$ and the unit vector in the $x^{(3)}$ direction, and $\phi$ is the azimuthal angle between the projection of $x$ onto the $x^{(3)}=0$ plane with the unit vector in the $x^{(1)}$ direction.  The relation between $x=x_{\theta,\phi}$ and $(\theta,\phi)$ can most easily be expressed according to the rule 
\[ x_{\theta,\phi}^{(1)} = \sin\theta \cos\phi, \quad x_{\theta,\phi}^{(2)} = \sin\theta\sin\phi, \quad x_{\theta,\phi}^{(3)}= \cos\theta. \]
It is easy to see that the Jacobian of this transformation is $(\theta,\phi) \mapsto \sin\theta$, so for any $\sigma$-integrable function $h$ on the sphere, integration is defined by 
\[ \int_\SS h(x) \, \sigma(dx) = \int_0^{2\pi} \int_0^\pi h(x_{\theta,\phi}) \sin\theta \, d\theta \, d\phi, \]
where the integration on the right-hand side is with respect to Lebesgue measure.  In what follows, when we define our kernels or mixing densities, we will use whichever is easier of the equivalent $x$ or $(\theta,\phi)$ definitions.  

%In the simplest notion of probability, if random variable $X$ has density function $f(x)$, the probability of $X$ occurring between $x$ and $x + dx$ is $f(x)dx$ and is also known as the ``probability density element'' over $dx$. Similarly for random variables $X$ and $Y$ in Euclidean space with joint pdf $f(x,y)$, the probability of $X$ and $Y$ between $x$ to $x+dx$ and $y$ to $y+dy$ respectively is $f(x,y)dxdy$, and this is because the measure in question is the product measure. For constructing a probability element on a sphere, the dynamics are a bit different because now the smallest element on a sphere is not simply a grid like $dxdy$. Let us first set up notation for spherical data.

Of course, the collection of all distributions on the sphere is virtually limitless, but not all of these would be natural choices for statistical models.  Two classes of distribution that have proved to be useful statistical models are the {\em vectorial} and {\em axial} distributions.  We discuss some axial and vectorial models below.

\subsubsection*{Vectorial distributions}

Most of the vectorial distributions defined in literature belong to the exponential family, and take the form $y \mapsto C e^{T(y,\eta)}$, $y \in \SS$, where $T$ is a given function, $\eta$ is a vector of parameters, and $C$ is the normalizing constant.  
%. In these distributions the probability density element is of the form $Ce^{T(x,\psi})dS^{2}$, where $x \in S^2$, $\psi$ is a vector of parameters and C is an appropriate integral constant.
von Mises--Fisher, \citet{kentfisher1982} and  \citet{wood1982} distributions are all applicable to vectorial data. The von Mises--Fisher distribution is an extension of the von Mises---or circular Gaussian---distribution on the unit circle to the unit sphere. It is characterized by two parameters: the vector location vector $\mu \in \SS$ and the scalar concentration parameter $\kappa > 0$.  The density is 
\begin{equation}
\label{eq:vMF}
k_\kappa(y \mid \mu) = \cfrac{\kappa^{1/2}}{(2\pi)^{3/2} I_{1/2}(\kappa)} \exp(\kappa \mu^\top y), \quad y \in \SS,
\end{equation}
where $\kappa \geq 0$ and $\mu \in \SS$; the normalizing constant, which depends on $\kappa$ but not $\mu$, is the modified Bessel function of the first kind, of order $1/2$.  This distribution is unimodal and rotationally symmetric about $\mu$, i.e., $y_{1}^\top \mu = y_{2}^\top\mu$ implies $k_\kappa(y_1 \mid \mu) = k_\kappa(y_2 \mid \mu)$. Because of its close ties to the normal distribution, the von Mises--Fisher distribution is one of the most commonly used to model spherical data.  

The Kent or the Fisher--Bingham distribution can be fit to a broader class of data than von Mises--Fisher. It is a five-parameter distribution, and can be fit to unimodal, bimodal or asymmetric data. 
%Since it can be fit to data having two modes, it can be used to model axial data as well. 
The Wood distribution was specially designed for bimodal data with equal probability in the two modes. It is also characterized by five parameters that define the modal directions and the angle between the two modes. This distribution is symmetric about the plane passing through the two modes and the origin. For a detailed overview of these and other distributions, see \citet{fisher1993} and \citet{mardia2009}.

% The relatively commonly used model is the von Mises Fisher distribution which is an analogue of the multivariate Normal distribution, wherein the observations are normalized vectors. We will talk more about this distribution in this section. The Kent distribution \cite{kentfisher1982} is more general in the sense that it is the analogue of the bivariate Normal distribution with an unconstrained covariance matrix. A class of girdle distributions was described in \cite{watson1965} and is usually referred to as the Watson distribution. All these distributions are circularly symmetric, i.e. they satisfy {Def above}. With respect to inference for the vector $\mu$, it is important to note that it can lie anywhere on the sphere $S^2$ without any identifiability issues for its estimation.

% \theoremstyle{definition}
% \begin{definition}{\textit{Definition:}}
% A distribution $F(y)$ on a sphere is said to be \textit{antipodally symmetric} if 

% $F(y) = F(-y) ~~ where ~~ y \in S^2$

% \end{definition}

\subsubsection*{Axial distributions}

These distributions are used to model data that \textit{do not} have a preferred direction and are referred to as antipodally symmetric, i.e., $f(y) = f(-y)$ for all $y \in \SS$.  \citet{watson1965}, \citet{bingham1974}, and \citet{tyler1987} are all examples of such models. The Watson distribution is the earliest attempt in literature to model such data. It is characterized by two parameters $\mu$ and $\kappa$ where, like von Mises--Fisher, $\mu$ is the location parameter and $\kappa$ is the shape parameter. However, the interpretation of $\mu$ differs based on whether $\kappa$ is positive or negative. If $\kappa \geq 0$, it is a {\em bipolar} model where observations are clustered at the poles of the axis $\pm \mu$. Conversely if $\kappa \leq 0$, it is a {\em girdle} model where observations are spread around on the ring of the circular plane passing through the origin, and normal to $\pm \mu$, now called the polar axis. This distribution is also rotationally symmetric. The Bingham distribution is a five-parameter distribution. In addition to a bipolar/girdle property it can also represent rotationally asymmetric data. For details about inference in these distributions see \citet{fisher1993} and \citet{mardia2009}. Another class of axial distributions are the angular central Gaussian (ACG) distributions. These are distributions for directions on an ellipse $y^\top \Sigma^{-1} y = 1$, where $\Sigma$ is the parameter.  ACG distributions are not exponential families, they have a simpler normalizing constant, and their statistical properties have been extensively studied \citep{tyler1987}. The Schladitz distribution introduced in \citet{schladitz2006} is a special case of ACG and one we use in this paper. The density function for the Schladitz distribution is given by 
% \begin{equation}
% \label{schladitz}
% k(\theta, \phi \mid (\theta_0, \phi_0), \beta) = \frac{1}{4\pi}\cfrac{\beta \sin{\theta}}{(1 + (\beta^2 - 1)(\cos(\phi_0 - \phi) \sin{\theta_0} \sin{\theta} + \cos{\theta_0} cos{\theta})^2)^{3/2}} \quad (\theta,\phi) \in [0,\pi) \times [0,2\pi)
% \end{equation}
\begin{equation}
\label{schladitz}
    k_\beta(y \mid \mu) = \cfrac{1}{4 \pi {|\Sigma_{\mu, \beta}|}^{1/2}}\cfrac{1}{(y^\top \Sigma_{\mu, \beta}^{-1} y)^{3/2}}, \quad y \in \SS
\end{equation}
where $\Sigma_{\mu,\beta} = Q_{\mu}^\top D_\beta Q_{\mu}$, with  $D_\beta = \text{diag}(1,1,\beta^{-2})$ and $Q_{\mu}$ is the rotation matrix mapping $(0,0,1)^\top$ onto the unit vector $\mu$, given by  
\begin{equation}
\label{eq:Q}
    Q_\mu = Q_{\theta_0,\phi_0} = \begin{pmatrix} \cos{\theta_0}\cos{\phi_0} & -\sin{\phi_0} & \sin{\theta_0}\cos{\phi_0} \\ 
                        \cos{\theta_0}\sin{\phi_0} & \cos{\phi_0} & \sin{\theta_0}\sin{\phi_0} \\
                        -\sin{\theta_0} & 0 & \cos{\theta_0}\\ 
                        \end{pmatrix},
\end{equation}
with $(\theta_0, \phi_0)$ the spherical coordinates of $\mu$.  This distribution is bipolar axial for $\beta < 1$, uniform over the sphere for $\beta = 1$, and girdle for $\beta > 1$.

\subsection{Mixture distributions on the sphere}
\label{mixtures}

% {\color{red} Bring some of the details from the current Section~4 here.  Also, make sure the notation is consistent. For example, if we're using $\psi$ for the mixing distribution, then, for the weights in a finite mixture, we should write $\psi_j = \psi(\mu_j)$, where $\mu_j$ is the location and $j=1,\ldots,J$ indexes the mixture components.}

The distributions mentioned in the previous section are sufficient to fit data from a single population or a population with one or possibly two modes. However, a data set might show multiple dominant directions, in which case the above distributions would be inappropriate models. In the past, several mixture models have been suggested for directional data. Given a kernel density function $k$ for observed variable $Y$ a finite mixture model with $J$ many components is represented by 
\[ f(y) = \sum_{j=1}^{J} \psi_{j} \, k_{\lambda}(y \mid \mu_j), \]
where $\mu_j$ is the location parameter for the $j^\text{th}$ component and,with a slight abuse of notation, the corresponding weight $\psi_j = \psi(\mu_j)$. The common non-mixing parameter $\lambda$ accounts for any structural features of the density, like $\kappa$ for von Mises--Fisher or $\beta$ for Schladitz.
\citet{franke2016} propose fitting such a discrete mixture of the Schladitz distribution to data from fiber composites. Here, estimation of parameters is achieved through use of EM algorithm for the ACG distribution of which Schladitz is a special case. \citet{banerjee2005} use the EM algorithm in clustering data on a hypersphere by fitting a mixture of the von Mises--Fisher distribution. \citet{peel2001} used mixtures of the Kent distribution, a general version of the von Mises--Fisher to fit a model for fracture directions obtained from a rock mass. The EM algorithm is used to estimate its parameters. More recently, \citet{zhang2013} discusses the fitting of mixture models to data in material sciences with emphasis on EM algorithm and maximum likelihood estimators. Section~\ref{simulations} considers both Schladitz and von Mises--Fisher mixtures.

The EM algorithm is frequently used for fitting finite mixtures on a sphere. It efficiently uses maximum likelihood estimation for the purposes of estimating the component parameters. Execution of the same starts with some initial values. Based on the complexity of the kernel, the maximum likelihood estimates can be highly sensitive to the user-supplied starting values for the EM algorithm; see Section~\ref{simulations}.  This creates a serious practical problem which, in addition to the well-known challenges related to EM's slow convergence and the inability for likelihood-based methods to estimate mixing densities, motivates a different approach.

\section{Predictive recursion on the sphere} 
\label{S:pr}

\subsection{The algorithm and its properties}
\label{SS:algorithm}

% The \textit{Predictive Recursion} (PR) algorithm is a non parametric recursive method that gives a smooth estimate of a mixing density. We describe the problem of estimation of a mixing density followed by the PR algorithm.Proposed in \cite{newton2002}, the PR estimator is a smooth non parametric recursive estimator of $p$ in \eqref{eq:mixture}. 

Suppose $Y^n = (Y_1,\ldots,Y_n)$ are independent and identically distributed observations on the sphere having mixture distribution as in \eqref{eq:mixture}. The goal is to estimate the unknown mixing distribution $\Psi$.  For this, the PR algorithm of \cite{newtonetal1998}, \citet{newtonzhang1999}, and \citet{newton2002} proceeds as follows.  

\begin{pralg}
Initialize the algorithm with a guess $\Psi_0$ of the mixing distribution and a sequence $\{w_i: i \geq 1\} \subset (0,1)$ of weights.  Then recursively update 
\begin{equation}
\label{eq:pr}
\Psi_i(dx) = (1-w_i) \, \Psi_{i-1}(dx) + w_i \, \frac{k(Y_i \mid x) \Psi_{i-1}(dx)}{f_{i-1}(Y_i)}, \quad i=1,\ldots,n, 
\end{equation}
where $f_{i-1}(y) = \int k(y \mid x) \, \Psi_{i-1}(dx)$ is the mixture corresponding to $\Psi_{i-1}$.  Return $\Psi_n$ and $f_n$ as the final estimates.  
\end{pralg}

%is estimation of $p$, where only $Y \in S^2$ are observable, while the distribution $p$ of the latent variable $X \in S^2$ is unknown.

%$i^{th}$ step in the PR algorithm is given by,
 
%\begin{equation}
%\label{eq:pr}
%p_i(x) = (1-w_i) \, p_{i-1}(x) + w_i \, \frac{k(Y_i \mid x)p_{i-1}(x)}{f_{i-1}(Y_i)}, \quad i=1,\ldots,n,
%\end{equation}

%where $f_{i-1}(y) = \int k(y \mid x)p_{i-1}(x) \, dS^2$ is the mixture corresponding to $p_{i-1}$.

For some intuition about what this algorithm is doing and why it works, note that the ratio on the right-hand side of \eqref{eq:pr} is the Bayesian update based on a single observation $Y_i$, with likelihood $k(Y_i \mid x)$, and prior distribution $\Psi_{i-1}$.  So the PR algorithm is simply computing, at each iteration, a weighted average of the ``prior'' and corresponding single-observation ``posterior'' distributions.  In \cite{newtonetal1998} it was shown that PR's one-step update, $(\Psi_{i-1}, Y_i) \mapsto \Psi_i$, corresponds exactly to the Bayesian posterior mean based on a Dirichlet process mixture model, where the ``prior'' for the mixing distribution is a Dirichlet process with base measure $\Psi_{i-1}$ and precision parameter $w_i$.  For more on the connection between PR and Dirichlet process mixture models, see Section~\ref{SS:prml}.  

Three practically relevant points about the PR algorithm are as follows.
\begin{itemize}
\item The PR algorithm is very fast to compute.  Specifically, since each data point requires a fixed set of computations, usually depending on the number of grid points at which the mixing density is to be evaluated (see below), the computational complexity is $O(n)$.  This computational speed gives the user the necessary flexibility to do some additional things with the PR estimate; see below on Section~\ref{SS:prml}.  
\vspace{-2mm}
\item The choice of weight sequence, $(w_i)$, affects PR's performance.  Two points are intuitively clear: first, if the weights do not vanish, then the updates will never stabilize and there is no hope for convergence; second, if they vanish too fast, then the iterations get stuck because of having insufficient opportunity to learn from data, hence no chance to be consistent. Theoretically, to achieve convergence we need the weights to vanish but not too quickly, which is a common problem in the stochastic approximation literature \citep[e.g.,][]{robbins1951, kushner2003}.  Specifically, the sequence must satisfy $w_i > 0$, $\sum_{i=1}^\infty w_i = \infty$, and $\sum_{i=1}^\infty w_i^2 < \infty$.  This is easy to arrange, and it is recommended to take $w_i = (i+1) ^{-\gamma}$, with $\gamma \in (1/2,1]$.  A slightly narrower range of $\gamma$ is suggested in \citet{martintokdar2009}.  In our numerical examples below, we use $\gamma = 2/3$.    
\vspace{-2mm}
\item It is easy to see that the final estimators, $\Psi_n$ and $f_n$, depend on the order in which the data $Y_1,\ldots,Y_n$ are processed.  Of course, in applications involving iid or exchangeable data, the order is irrelevant (i.e., an ancillary statistic), so this order-dependence could be seen as an undesirable property.  To at least partially eliminate this order-dependence, it has been recommended \citep[e.g.,][]{newton2002, tmg} to average the order-dependent PR estimates over a number of permutations of the data sequence, i.e., a Rao--Blackwellization step.  This is feasible thanks to the speed of an individual run of PR.  More recently, it has been proposed to leverage PR's order-dependence for the purpose of deriving hypothesis tests and confidence intervals relevant to $\Psi$ \citep{dixitmartin2019}.  
\end{itemize}

Asymptotic convergence properties of the PR estimators, $\Psi_n$ and $f_n$, were first investigated in \citet{newton2002}, \cite{ghoshtokdar2006}, and \cite{martinghosh2008}, and then later refined and extended in \citet{tmg} and \cite{martintokdar2009}.  Of particular importance in the convergence theory for PR---and for many other methods that aim to solve this kind of inverse problem---is compactness of the support of the mixing distribution.  In many applications, this assumption fails without some artificial or impractical constraints.  However, in the present case where interest is in mixing distributions supported on the unit sphere in $\RR^3$, compactness is immediate. And given compactness, it is relatively easy to show that the kernel densities of interest here, described in Section~\ref{SS:sphere.dist} above, satisfy the conditions of the general PR convergence theory; see Propositions~1 and 2 in the Appendix. Below is a summary of the main conclusions, with more precise details given in the Appendix.  

\begin{thm}
\label{thm:pr}
If $f^\star$ is the true density for the iid data $Y_1,Y_2,\ldots$, which may or may not be of the mixture form \eqref{eq:mixture}, and if the conditions stated in the Appendix are satisfied, then the following properties of the PR estimators, $f_n$ and $\Psi_n$, hold:
\begin{enumerate}
\item The PR estimator of the mixture density, $f_n$, satisfies 
\[ K(f^\star, f_n) \to \inf K(f^\star, f), \quad \text{with $f^\star$-probability 1, as $n \to \infty$}, \]
where $K(f^\star, f)$ is the Kullback--Leibler divergence of $f$ from $f^\star$, and the infimum is over all $f$ of the mixture form \eqref{eq:mixture}.  If the mixture model is correctly specified, so that $f^\star$ is of the form \eqref{eq:mixture} for some $\Psi$, then the infimum is 0 and $f_n$ is a consistent estimator of $f^\star$.  
\vspace{-2mm}
\item If the posited mixture model \eqref{eq:mixture} is identifiable, then $\Psi_n \to \Psi^\star$ weakly, with $f^\star$-probability~1, as $n \to \infty$, where $\Psi^\star$ is the unique mixing distribution corresponding to the infimum in Part~1 above.  If the mixture model is correctly specified, then $\Psi^\star$ is the true mixing distribution and $\Psi_n$ is a consistent estimator of $\Psi^\star$.
\end{enumerate}
\end{thm}

Part~2 of Theorem~\ref{thm:pr} makes a non-trivial identifiability assumption about $\Psi$ in the mixture model \eqref{eq:mixture}.  For deconvolution problems on the real line, identifiability is relatively straightforward but, beyond that, the question becomes quite difficult.  Classical references on identifiability for mixing distributions supported on Euclidean spaces include \citet{teicher1961, teicher1963}.  To our knowledge, identifiability results for mixture models on the sphere are limited to \citet{kent1983}, who shows that finite mixtures of exponential family kernels on the sphere are identifiable.  Identifiability for general mixtures on the sphere is an open problem; see Section~\ref{S:discuss}.

%An important assumption for consistency and rate of convergence is the compactness of the support of the mixing density. Since our mixing density lies on a sphere or a part of the sphere, this is trivially satisfied. The mixture estimate $f_n$ obtained through $\tilde{p}_n$ is consistent for $f$. Thus we have an asymptotically unbiased estimator for the mixture $f$. This is backed by simulation results given in section \ref{simulations}. Moreover, rate of convergence results have been derived for a finite mixture., as given in the Appendix. Finally, even for a misspecified mixture (in terms of the structural parameter $\theta$ of the kernel distribution), the PR estimate $\tilde{p}_n$ is consistent in the weak topology.

%   We show some theoretical results in the appendix about the consistency of $\tilde{p_n}$ and the consistency and rate of convergence of the resulting $\tilde{f_n}$, estimate of the mixture.

% Before we jump into the implications of the Propositions, let us justify the assumptions.  

% Proposition 1 claims consistency of the estimate of the mixture resulting from the PR estimate $\tilde{p_n}$.  For the weight choices given in Proposition 1.B, the rate of convergence of $\tilde{f_n}$ is of the order of $o(n^{-(1-\gamma)/2})$. Proposition 2 is more general in the sense that even if the model is misspecified in terms of the structural parameters $\theta$ of the kernel distribution, the PR estimate is consistent in the weak topology

\subsection{Computational details}
\label{SS:prcomp}

Here we discuss some computational details relevant to implementation of the PR algorithm in the spherical context.  First, a general mixing measure $\Psi$ is too complex so, for computational tractability, it is common to impose some structure on $\Psi$.  On one hand, the likelihood-based methods described above automatically impose a discrete structure, which may be undesirable.  The PR algorithm, on the other hand, is more flexible, allowing the user to specify a dominating measure with respect to which $\Psi$ has a density.  In our present context, a natural candidate for this dominating measure is the surface area measure $\sigma$.  With this structure, the mixture model \eqref{eq:mixture} can be re-expressed as 
\[ f(y) = \int k(y \mid x) \, \psi(x) \, \sigma(dx), \]
where $\psi = d\Psi/d\sigma$ is the unknown density of $\Psi$ with respect to $\sigma$.  Similarly, the PR algorithm can be modified to estimate the density directly, so the update \eqref{eq:pr} can also be re-expressed as 
\[ \psi_i(x) = (1-w_i) \, \psi_{i-1}(x) + w_i \frac{k(Y_i \mid x) \, \psi_{i-1}(x)}{f_{i-1}(Y_i)}, \quad i=1,\ldots,n. \]

Second, a quadrature method is needed to numerically evaluate the integral in the normalizing constant $f_{i-1}(Y_i)$ above and, more generally, to estimate the mixture density $f$.  For this, we suggest two approaches. The naive approach is to convert the variable $x$ on the sphere, represented in terms of Cartesian coordinates, to the spherical coordinates $(\theta, \phi)$ on the rectangle $[0,\pi) \times [0,2\pi)$.  As described in Section~\ref{SS:sphere.dist} above, if $\psi$ is a density on the sphere, then a corresponding density $\psi^\dagger$ on the rectangle is given by 
\begin{equation}
\label{eq:change.of.variables}
\psi^\dagger(\theta,\phi) = \psi(x_{\theta,\phi}) \, \sin\theta, \quad (\theta,\phi) \in [0,\pi) \times [0,2\pi), 
\end{equation}
where $x_{\theta,\phi}$ is the spherical-to-Cartesian coordinate transformation.  In some applications, especially if the kernel is antipodal, the support of the mixing distribution may not be the entire sphere so, of course, the range of $(\theta,\phi)$ would be suitably restricted in the above expression and in what follows.  Then the mixture density can be expressed in terms of an integration with $\psi^\dagger$ over the rectangle, i.e., 
\begin{equation}
\label{eq:reparamixture}
f(y) = \int_0^{2\pi} \int_0^\pi k(y \mid x_{\theta,\phi}) \, \psi^\dagger(\theta, \phi) \, d\theta \, d\phi, \quad y \in \SS. 
\end{equation}
Now that the to-be-evaluated normalizing constant can be expressed as numerical integration over a two-dimensional grid, the computations are relatively straightforward. Our spherical PR algorithm proceeds as follows. We start with an initial estimate of the density $\psi_0^\dagger$ supported on the grid $[0,\pi) \times [0, 2\pi)$.  Then we evaluate $\psi_i^\dagger$ along the data sequence, $i=1,\ldots,n$, computing the normalizing constant $f_{i-1}(Y_i)$ using the quadrature rule described above.  The final mixture density estimate, $f_n$, can be evaluated directly from $\psi_n^\dagger$ according to formula \eqref{eq:reparamixture}, and to evaluate the mixing density estimator, $\psi_n$, on the sphere we simply invert the change-of-variables in \eqref{eq:change.of.variables}. 

A less-naive approach is based on Lebedev quadrature proposed in \citet{lebedev1976}. It gives an approximation for the surface integration over a sphere in three dimensions. The procedure is to form a grid to accurately integrate a polynomial function of degree $p$ on the sphere.  Hence, computational complexity can be controlled by adjusting the number of grid points based on the complexity of the function to be integrated. A low degree polynomial could be integrated accurately with a small number of grid points.  However, our mixing density $\psi$ is unknown, so we cannot be sure that $k(Y_i \mid x) \, \psi(x)$ would be simple enough to be accurately integrated over a small number of grid points.  If we take a much finer set of grid points, to protect against inaccurate integration of a rough function, then our benefit to using the efficient Lebedev spacing is limited.  Therefore, we opt to use quadrature with the simpler naive grid points described above.  

%however, when integrating over the mixing density in the mixture $f(y)$, $\psi(x)$ could be a smooth function or a discrete distribution highly concentrated at a couple of points. In the latter case we need a finer grid to account for its structure. Since the mixing density is unknown we would need this fine grid to account for the worst-case scenario. This entails the computation time to be the same as the former method, hence we stick to the first approach.

% Based on the degree a suitable number of grid points can be chosen to approximate the integration.The grid points are constructed such that they have octahedral symmetry over the sphere.The available dataset $Sphere\_Lebedev\_Rule$ provides the Lebedev grid points for a precision $p$. The PR procedure is similar to the one described in the first case but simpler in the sense that it doesn't involve the change of variables of $x_{\theta, \phi}$ to $(\theta,\phi)$.%

Finally, recall that the PR estimator depends on the order in which the data sequence is processed.  To remove this order-dependence, the PR algorithm described above can be repeated for any number of permutations and the final estimators would be point wise averages of the permutation-specific PR estimators.  That is, let $\S_n$ be the set of all permutations $\{1, 2, \ldots, n\}$, and define $\psi_n^{(s)}$ to be the PR estimator based on data $Y^n$ arranged according to permutation $s \in \S_n$.  Then one can construct the (``Rao--Blackwellized'') {\em permutation-averaged} PR estimator 
\[ \psi_n^{\text{perm}}(x) = \frac{1}{n!} \sum_{s \in \S_n} \psi_n^{(s)}(x). \]
Of course, computing this {\em for all permutations} would be a burden, especially for large $n$, but a Monte Carlo approximation can be readily obtained by averaging only over, say, 25 permutations sampled at random from $\S_n$.

\subsection{Marginal likelihood}
\label{SS:prml}

The PR algorithm is designed for cases where the only unknown in the mixture model is the mixing distribution itself.  However, it is often the case that there are other unknown {\em structural} or {\em non-mixing} parameters in the mixture model.  For example, mixture models are commonly used for density estimation but with an unknown bandwidth parameter to provide the necessary flexibility.  So consider the following generalization of the original mixture model \eqref{eq:mixture} 
\begin{equation}
\label{eq:mixture2}
f(y) = f_{\lambda,\Psi}(y) = \int k_\lambda(y \mid x) \, \Psi(dx), 
\end{equation}
where $\Psi$ is still the unknown mixing distribution but now the kernel $k=k_\lambda$ depends on an unknown structural parameter $\lambda$ taking values in $\Lambda \subseteq \RR^d$ for some $d \geq 1$.  The goal now is to estimate both $\Psi$ and $\lambda$.  

For iid data $Y_1,\ldots,Y_n$ from the extended mixture model \eqref{eq:mixture2}, the likelihood is 
\begin{equation}
\label{eq:lik}
(\lambda, \Psi) \mapsto L_n(\lambda, \Psi) = \prod_{i=1}^n \int k_\lambda(Y_i \mid x) \, \Psi(dx). 
\end{equation}
A reasonable strategy for modifying the PR algorithm to estimate $\lambda$ would be to use a (fixed-$\lambda$) PR estimator of $\Psi$ to somehow eliminate the likelihood's dependence on $\Psi$, leaving a sort of likelihood only in $\lambda$ that could be maximized to produce an estimate.  \citet{tao1999} proposed to run the PR algorithm at each fixed $\lambda$ value to produce a $\lambda$-dependent estimator $\Psi_{n,\lambda}$ and then define a {\em profile likelihood} as 
\[ \lambda \mapsto L_n^{\text{\sc p}}(\lambda) = L_n(\lambda, \Psi_{n,\lambda}), \]
to be maximized over $\lambda \in \Lambda$.  A general issue with the use of profile likelihoods is that it effectively ignores the uncertainty in the plug-in estimator.  Motivated by the PR algorithm's close connection to the Dirichlet process mixture models, \citet{martintokdar2011} propose an alternative {\em marginal likelihood}
\begin{equation}
\label{eq:prml}
\lambda \mapsto L_n^{\text{\sc m}}(\lambda) = \prod_{i=1}^n \int k_\lambda(Y_i \mid x) \, \Psi_{i-1,\lambda}(dx), 
\end{equation}
which, again, can be maximized over $\lambda \in \Lambda$ to produce an estimator, i.e., 
\begin{equation}
\label{eq:prmle}
\hat\lambda_n = \arg\max_{\lambda \in \Lambda} L_n^{\text{\sc m}}(\lambda). 
\end{equation}
Computationally, evaluation of $\hat\lambda_n$ requires multiple runs of the PR algorithm, one for each iteration in the optimization procedure.  Fortunately, each PR run is fast, so maximizing the PR marginal likelihood is relatively efficient.  

Questions about the general statistical properties (e.g., consistency, asymptotic normality, etc.)~remain open, but here we provide some heuristic theoretical support.  Let $f^\star$ denote the true density, not necessarily of the form \eqref{eq:mixture2}, and define the function 
\[ K^\star(\lambda) = \textstyle\inf_\Psi K(f^\star, f_{\lambda, \Psi}), \]
where the infimum is over all mixing distributions $\Psi$ as in \eqref{eq:mixture2}.  The general PR convergence theory in Part~1 of Theorem~\ref{thm:pr} implies that, if $f_{n,\lambda}$ is the $\lambda$-specific PR estimator, then $K(f^\star, f_{n,\lambda})$ converges almost surely, as $n \to \infty$, to $K^\star(\lambda)$ for all fixed $\lambda$; note that the ``for all $\lambda$'' conclusion requires the full force of Theorem~\ref{thm:pr}, i.e., the convergence under a misspecified model, since other values besides the ``true'' $\lambda$ are being considered.  An empirical version of the Kullback--Leibler divergence, 
\[ K_n(\lambda) = \frac1n \sum_{i=1}^n \log \frac{f^\star(Y_i)}{f_{i-1,\lambda}(Y_i)}, \]
was proposed in \citet{martintokdar2011} and shown, in their Theorem~2, to converge pointwise to $K^\star$:
\begin{equation}
\label{eq:pointwise}
K_n(\lambda) \to K^\star(\lambda), \quad \text{with $f^\star$-probability~1, as $n \to \infty$, for all fixed $\lambda \in \Xi$}. 
\end{equation}
Since maximizing $L_n^{\text{\sc m}}(\lambda)$ in \eqref{eq:prmle} is equivalent to minimizing $K_n$, in light of \eqref{eq:pointwise}, a reasonable conjecture is that ``$\inf_\lambda K_n(\lambda) \to \inf_\lambda K^\star(\lambda)$'' in some sense, which, in turn, suggests that $\hat\lambda_n$ in \eqref{eq:prmle} converges to $\arg\min_\lambda K^\star(\lambda)$.  Of course, if the model \eqref{eq:mixture2} is correctly specified, then the latter arg-min equals the true value of $\lambda$, in which case $\hat\lambda_n$ would be a consistent estimator.  The challenge to making these latter claims rigorous is establishing a uniform-in-$\lambda$ version of \eqref{eq:pointwise}; see Section~\ref{S:discuss}.  Despite these unanswered questions, the above heuristics and the strong empirical performance shown in \citet{martintokdar2011, martintokdar2012}, \citet{prml-finite}, and \citet{martinhan2016} suggest that this PR marginal likelihood-based methodology is quite useful.  In Section~\ref{goodnessoffit}, we use this marginal likelihood framework to develop a novel goodness-of-fit test to compare a parametric model (i.e., a one-component mixture) on the sphere to a mixture thereof.

\section{Simulations}
\label{simulations}

\subsection{Overview}

% {\color{red} What I want this overview section to say is the following:
% \begin{itemize}
% \item what methods are being compared and any specific points about them, e.g., the EM's dependence on starting values, permutations in PR and PRML, how are we optimizing PRML, etc; 
% \item what models (kernels and mixing distributions) are being considered;
% \item what features are we interested in and what metrics we use to compare the performance of methods relative to these features.
% \end{itemize}
% All the definitions (e.g., density functions below) should be given in previous sections, and all statements about the simulation results should be in the subsections below.} Resolved.

Here we carry out a simulation study to investigate the performance of the mixing density estimation methods described in Section~\ref{S:pr}. Data is generated from a mixture distribution of the form \eqref{eq:mixture2} using different kernels and mixing distributions. As discussed in Section~\ref{SS:sphere.dist}, directional data can be divided in two main categories; axial or vectorial. To account for the effect of this structure, we use kernels of both types namely Schladitz (axial) and von Mises--Fisher (vectorial). For each kernel, we consider a variety of finite and continuous mixing distributions.

% \begin{enumerate}
% \item $\psi(\theta, \phi) = \psi \, \delta_{(\pi/2, 0)} + (1-\psi) \, \delta_{(0,0)}$, where $\psi$ denotes the mixing proportion and $\delta_x$ the Dirac point-mass distribution at $x$; 
% \item $\psi(\theta, \phi) = \psi \, \delta_{(\pi/2, 0)} + (1-\psi) \, \delta_{(\pi/2,\pi/2)}$;
% \item $\psi(\theta, \phi) = \mathsf{trN}_2 (\eta, \Sigma)$, where $\mathsf{trN}_2$ denotes a bivariate normal distribution with mean vector and covariance matrix 
% \[ \eta = (\pi/4, \pi)^\top \quad \text{and} \quad \Sigma = \begin{pmatrix} (\pi/12)^2 & (\pi/12)^2 \\ (\pi/12)^2 & (\pi/3)^2 \end{pmatrix}, \]
% truncated to the $(\theta,\phi)$ range; 
% \item $\psi(\theta,\phi) = \bet(2,5) \times \bet(2,2)$; \note{rescaled!}
% \item $\psi(\theta,\phi) = 0.5 \, \mathsf{trN}_2(\eta_1, \Sigma) + 0.5 \,  \mathsf{trN}_2(\eta_2, \Sigma)$, where the mean vectors are $\eta_1 = (\pi/4, \pi/2)^\top$ and $\eta_2 = (\pi/4, 5\pi/4)^\top$, and the covariance matrix is
% \[ \Sigma = \begin{pmatrix} (\pi/12)^2 & 0 \\ 0 & (\pi/6)^2 \end{pmatrix}. \]
% \item $\psi(\theta,\phi) = \bet(4,4) \times \unif(0,1)$ \note{rescaled!} {\color{red}(case 1 : $\phi \sim$ Unif and case 2 : $\theta \sim$ Unif)}               
% \end{enumerate}

The comparisons are between the PR estimator as described above and the maximum likelihood estimator as implemented via the EM algorithm.  For the PR implementation, we first estimate the structural parameter using the approach proposed in Section~\ref{SS:prml}. Then given this estimate, the PR estimate of the mixing density is obtained using the method described in Section~\ref{S:pr}. For removing PR's order-dependence, we average over 10 random permutations. An estimate of the mixture $f(y)$ is obtained using this estimate. Implementation of the EM algorithm is based on fitting a finite mixture as stated in Section~\ref{mixtures}, where the common structural parameter and location parameter for each component $j$ are estimated in an iterative procedure by maximizing the updated expected likelihood. We allow EM algorithm to choose up to 10 mixture components, with the choice being made based on BIC. We use the EM setup suggested in \citet{franke2016} when using the Schladitz kernel, and the setup by \citet{hornik2014} for the von Mises--Fisher kernel. Indeed, performance of this algorithm is affected by the structure of the kernel density, the true mixing distribution and choice of initial values, as can be seen in the Schladitz case. Also, with reference to \citet{ng2020}, the likelihood function for finite mixture models is often unbounded, so maximizing the likelihood is not well-defined implying that the EM might only find a local maximum. 

For the comparisons, we use two different metrics.  The first is the Kullback--Leibler divergence between mixture densities.  The second is to compare the mixing distribution estimates but, since we are sometimes comparing discrete and continuous distributions, we opt for a version of the $L_1$ distance between discretized versions.  That is, we partition the $(\theta,\phi)$ range into rectangles $A_k$, for $k=1,\ldots,K$ and, for an estimate $\widehat\Psi^\dagger$ of the true distribution $\Psi^\dagger$ on the $(\theta,\phi)$-space, we compute 
\begin{equation}
\label{eq:mixing_calculation}
d(\Psi^\dagger, \widehat\Psi^\dagger) = \sum_{k=1}^K \bigl| \widehat\Psi^\dagger(A_k) - \Psi^\dagger(A_k) \bigr|.
\end{equation}
For each case, we generate samples of size $n=2000$ from the stated mixture model, and the results presented below are based on 50 replications.  

%Broadly we are interested in efficiency of PR in estimation of the mixing and mixture distributions, compared with maximum likelihood.  However, and effect of initial values on estimation using EM. Note that unlike EM, the main goal of PR is not estimation of the mixture but the mixing density, hence essentially one comparison shows the effectiveness of the byproduct of PR. The simulations also answer the question of whether the structure of the mixing distribution affects its estimation via PR. Below we split our simulation results into those for Schladitz and von Mises--Fisher mixtures.

\subsection{Schladitz mixtures} 
\label{Schladitz}

Since this is an antipodally symmetric distribution, to avoid identifiability issues we use only the upper hemisphere as support for the probability density of the location parameter $\mu$, i.e, for $(\theta_0, \phi_0)$ as the spherical coordinates of $\mu$, $(\theta_0,\phi_0) \in \{(0,\pi/2) \times (0,2\pi)\}$. The mixing distributions we use for the simulations are as follows:

\begin{enumerate}
\item {\em Two-point discrete distribution:} $\psi(\theta, \phi) = \psi \, \delta_{(\pi/2, 0)} + (1-\psi) \, \delta_{(0,0)}$, where $\psi$ denotes the mixing proportion and $\delta_x$ the Dirac point-mass distribution at $x$.  Cases (a)--(d) correspond to $\psi = 0.5, 0.25, 0.2, 0.1$, respectively.

\item {\em Unimodal continuous distribution:} $\psi(\theta, \phi) = \mathsf{trN}_2 (\eta, \Sigma)$, where $\mathsf{trN}_2$ denotes a bivariate normal distribution with mean vector and covariance matrix 
\[ \eta = (\pi/4, \pi)^\top \quad \text{and} \quad \Sigma = \begin{pmatrix} (\pi/12)^2 & (\pi/12)^2 \\ (\pi/12)^2 & (\pi/3)^2 \end{pmatrix}, \]
truncated to the $(\theta,\phi)$ range; 

\item {\em Unimodal continuous distribution:} $\psi(\theta,\phi) = \frac{2}{\pi}\bet(\frac{2\theta}{\pi} \mid 2,5) \times \frac{1}{2\pi}\bet(\frac{\phi}{2\pi} \mid 2,2)$.  This has a larger spread than the truncated normal above.

\item {\em Bimodal continuous distribution:} $\psi(\theta,\phi) = 0.5 \, \mathsf{trN}_2(\eta_1, \Sigma) + 0.5 \,  \mathsf{trN}_2(\eta_2, \Sigma)$, where the mean vectors are $\eta_1 = (\pi/4, \pi/2)^\top$ and $\eta_2 = (\pi/4, 5\pi/4)^\top$, and the covariance matrix is
\[ \Sigma = \begin{pmatrix} (\pi/12)^2 & 0 \\ 0 & (\pi/6)^2 \end{pmatrix}. \]
\end{enumerate}

The structural parameter $\beta$ is fixed to be 0.1 over all components. Simulation results are given in Table~\ref{tab:sch}. Let us first look at the results for estimation of mixing distribution, as measured by the distance $d$ in \eqref{eq:mixing_calculation}. When the truth is a two point discrete distribution, the discrete estimate obtained via EM does slightly better then PR. Understandably in the continuous case, the continuous estimate obtained via PR outperforms the estimate via EM for the true mixing density. The Kullback--Leibler divergence results for estimating the mixture density convey more interesting insights. For the first four finite mixture cases one would expect efficient estimation via EM. However, both the mean and standard errors are larger for EM than PR. We believe this is due to the choice of initial values for the location parameter $\mu$. As we can see from the bimodal histogram in Panel~(a) of Figure~\ref{fig:em_hist}, even though a large number of the Kullback--Leibler divergence  values are small, a significant number exhibit behavior away from these values. For an empirical check, we ran the EM algorithm with the true means as initial values, which resulted in fewer large Kullback--Leibler values and hence better performance. For the case of 1(c), we can see that the mean Kullback--Leibler divergence is smaller as it has fewer large values than case (a). Moreover, despite lacking the knowledge that the mixture is finite, PR does a reasonable good job in terms of estimating the mixture density. As expected, the PR estimator outperforms the maximum likelihood estimator in cases where the true mixing density is continuous.  For visualization purposes, Figure~\ref{fig:mix_norm} shows contour plots of the true mixing distribution---the mixture of normals case in Example~4---along with the PR estimate in the case of a Schladitz kernel.  Clearly, PR is picking up the overall shape of the true mixing density quite well.

\begin{table}[t]
\centering
\begin{tabular}{ccccc}
    \hline
Mixing density & $\text{KL}_{\text{PR}}$ & $\text{KL}_{\text{ML}}$ & $d(\Psi^\dagger, \widehat\Psi^\dagger)_{\text{PR}}$ & $d(\Psi^\dagger, \widehat\Psi^\dagger)_{\text{ML}}$ \\
 %& mean(se) & mean(se) & mean(se) & mean(se)\\
 \hline
 1(a) & 0.014 (0.001) & 0.098 (0.023) & 0.600 (0.006) & 0.305 (0.024)\\
 %\hline
 1(b) & 0.012 (0.001) & 0.044 (0.017) & 0.533 (0.008) & 0.262 (0.021)\\
 %\hline
 1(c) & 0.011 (0.001) & 0.029 (0.010) & 0.504 (0.009) & 0.201 (0.014)\\
 %\hline
 1(d) & 0.010 (0.001) & 0.043 (0.003) & 0.477 (0.009) & 0.132 (0.016)\\
 %\hline
 2 & 0.017 (0.001) & 0.034 (0.001) & 0.300 (0.010) & 1.806 (0.025)\\
 %\hline
 3 & 0.013 (0.001) & 0.028 (0.000) & 0.396 (0.009) & 2.046 (0.028)\\
 %\hline
 4 & 0.021 (0.003) & 0.022 (0.001) & 0.264 (0.009) & 2.110 (0.026)\\
 \hline
\end{tabular}
\caption{Simulation results using the Schladitz distribution as the kernel density.  Entries correspond to the mean of each feature, with standard errors in parentheses. KL denotes the Kullback--Leibler divergence of the mixture density estimate from the truth.}
\label{tab:sch}
\end{table}

\begin{figure}
    \centering
    \subfloat[]{{\includegraphics[width = 6cm, height = 6 cm]{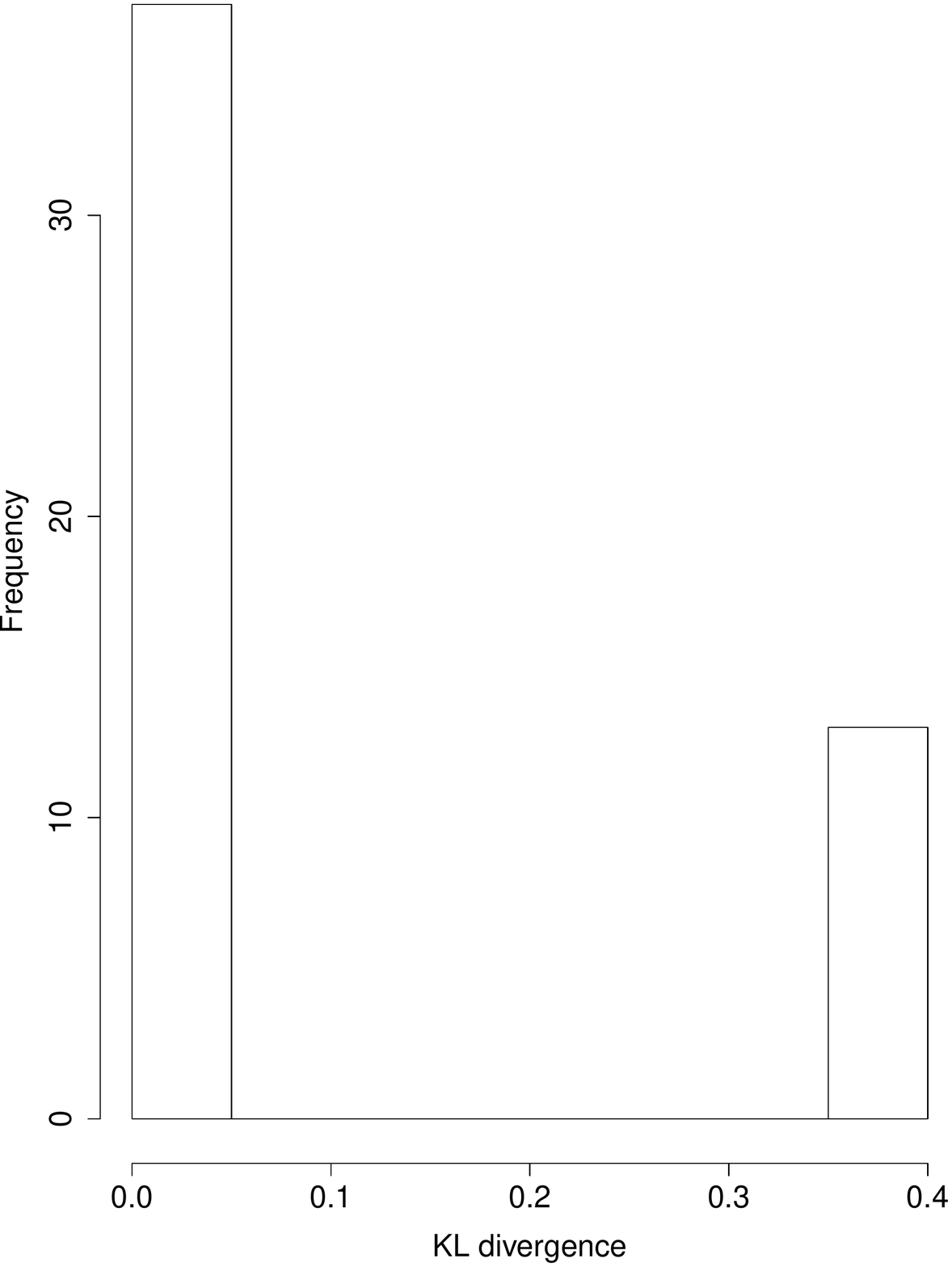}}}%
    \qquad
    \subfloat[]{{\includegraphics[width = 6cm, height = 6 cm]{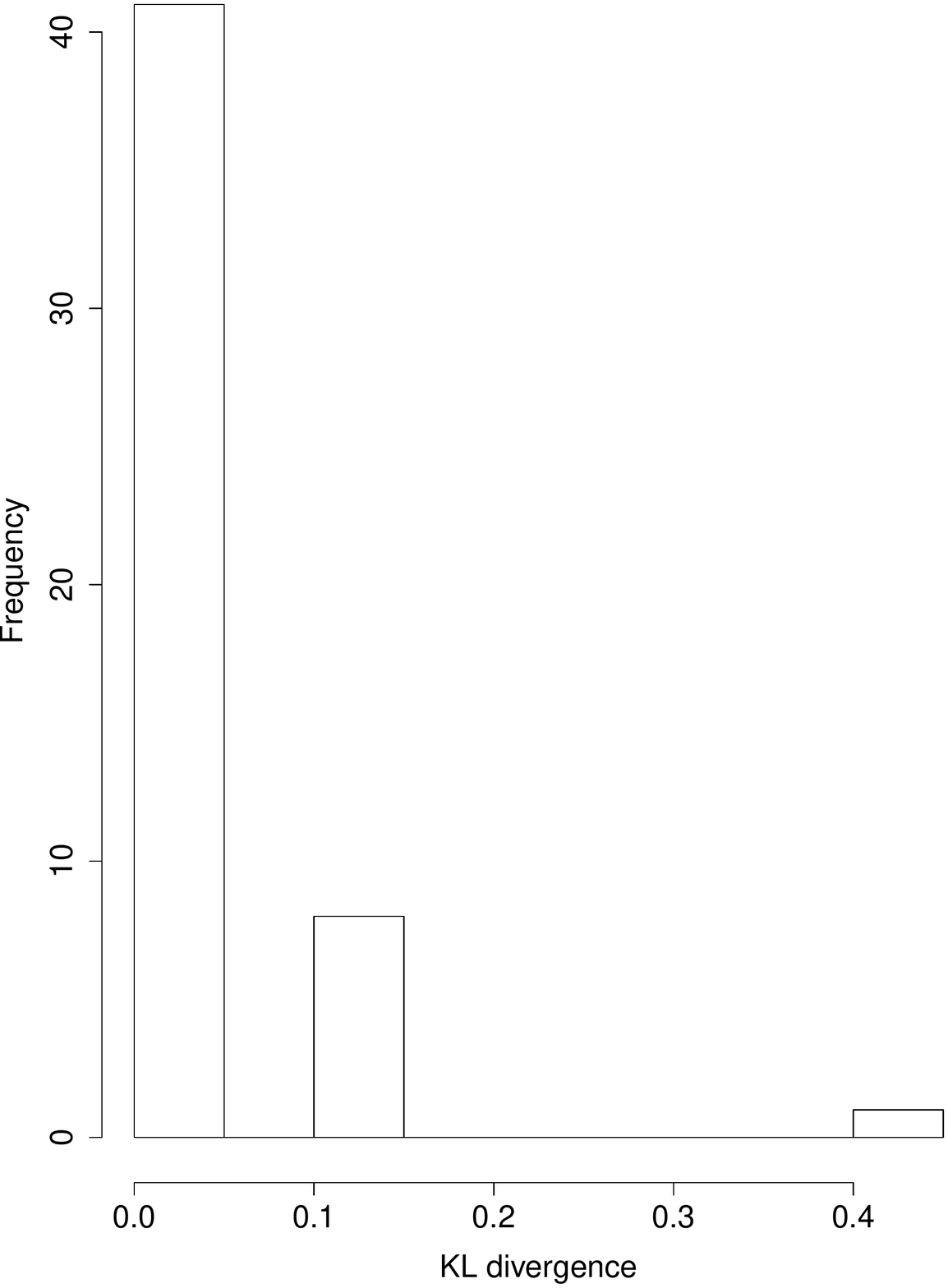}}}%
    \caption{Histogram of the Kullback--Leibler divergence values obtained from maximum likelihood estimation of mixing density 1(a) and 1(c) respectively, using the EM algorithm under the Schladitz kernel.}  %
    \label{fig:em_hist}
    \end{figure}

% How to conclude at the end, what more can we add here? Should we add a simulation with the bivariate normal density as the mixing density as a standalone example of PR?

\begin{figure}[t]
\begin{center}
\subfloat[Mixing density 4 on a 2-dimensional grid]{
  \includegraphics[width=8 cm]{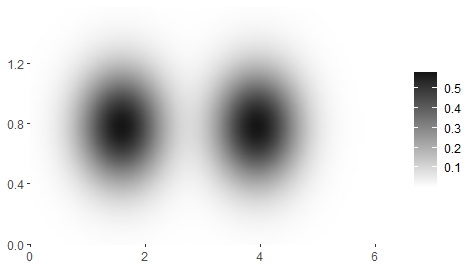}}
\subfloat[Mixing density 4 induced on a hemisphere]{
  \includegraphics[width=8 cm]{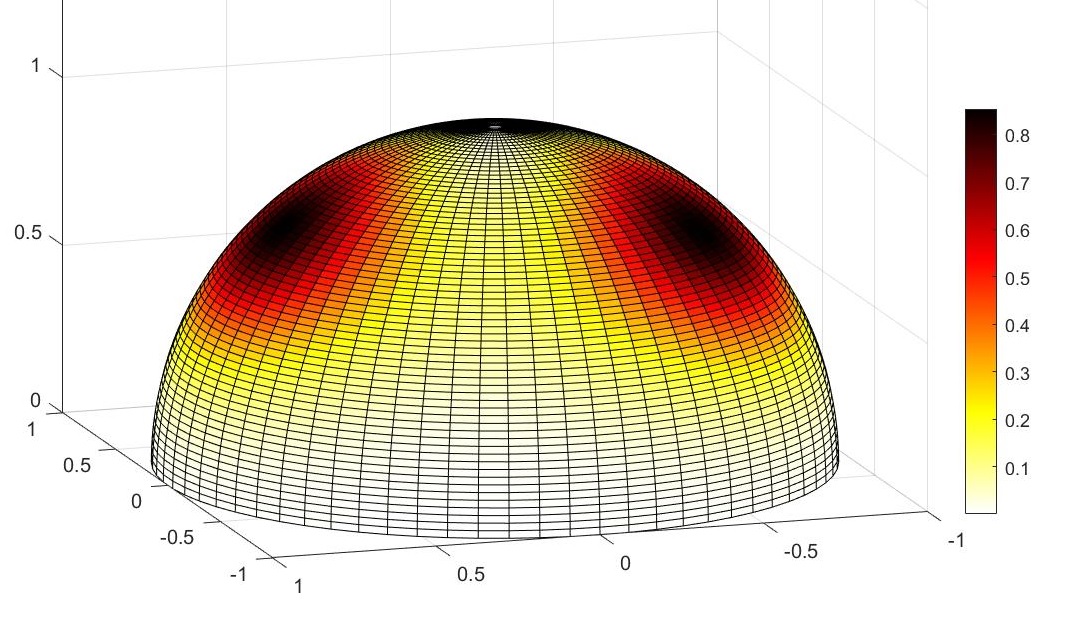}}
\hspace{0mm}
\subfloat[PR estimate of mixing density 4]{
  \includegraphics[width=8 cm]{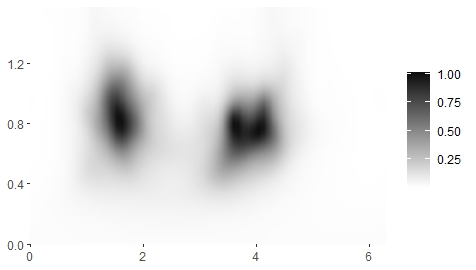}}
\subfloat[PR estimate of mixing density 4, induced on a hemisphere]{
  \includegraphics[width=8 cm]{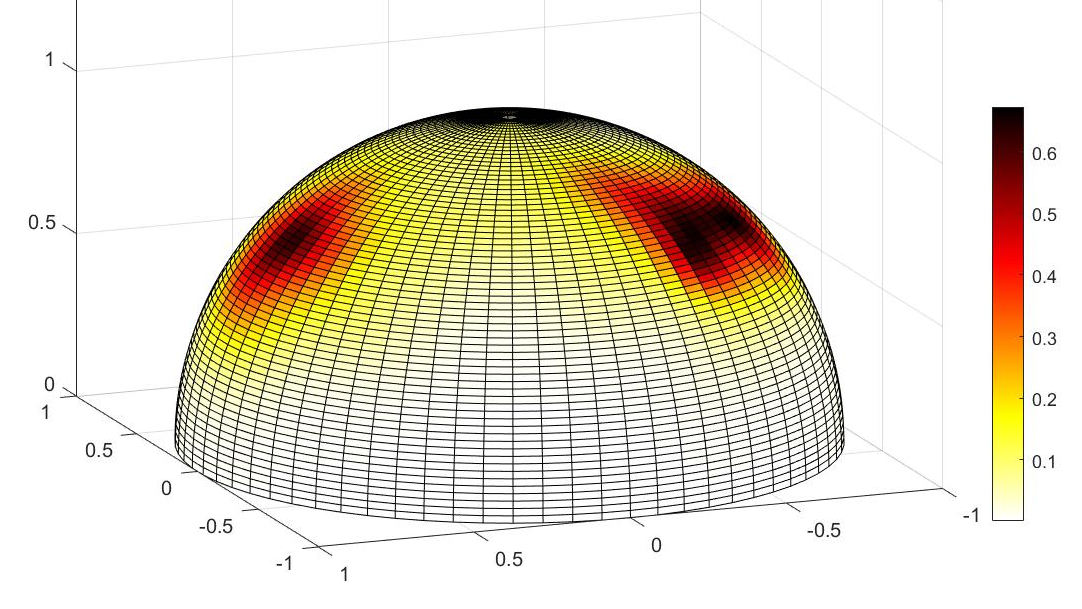}}
\end{center}
\caption{}
\label{fig:mix_norm}
\end{figure}

\subsection{von Mises--Fisher mixtures}
\label{vMF}

Since the von Mises--Fisher is a vectorial distribution, the support for the mixing distribution is now the entire sphere, i.e, for $(\theta_0, \phi_0)$ as the spherical coordinates of $\mu$, $(\theta_0,\phi_0) \in \{(0,\pi) \times (0,2\pi)\}$. The mixing distributions we use for the von Mises--Fisher simulations are as follows:

\begin{enumerate}
\item {\em Two-point discrete distribution:} $\psi(\theta, \phi) = 0.5 \, \delta_{(\pi/2, 0)} + 0.5 \, \delta_{(\pi/2,\pi/2)}$. 

\item {\em Unimodal continuous distribution:} $\psi(\theta, \phi) = \mathsf{trN}_2 (\eta, \Sigma)$, where $\mathsf{trN}_2$ denotes a bivariate normal distribution with mean vector and covariance matrix 
\[ \eta = (\pi/4, \pi)^\top \quad \text{and} \quad \Sigma = \begin{pmatrix} (\pi/12)^2 & (\pi/12)^2 \\ (\pi/12)^2 & (\pi/3)^2 \end{pmatrix}. \]

\item {\em Unimodal continuous distribution:} $\psi(\theta,\phi) = \frac{1}{\pi}\bet(\frac{\theta}{\pi} \mid 2,5) \times \frac{1}{2\pi}\bet(\frac{\phi}{2\pi} \mid 2,2)$.  This has a larger spread than the truncated normal above.

\item {\em Bimodal continuous distribution:} $\psi(\theta,\phi) = 0.5 \, \mathsf{trN}_2(\eta_1, \Sigma) + 0.5 \,  \mathsf{trN}_2(\eta_2, \Sigma)$, where the mean vectors are $\eta_1 = (\pi/4, \pi/2)^\top$ and $\eta_2 = (\pi/4, 5\pi/4)^\top$, and the covariance matrix is
\[ \Sigma = \begin{pmatrix} (\pi/12)^2 & 0 \\ 0 & (\pi/6)^2 \end{pmatrix}. \]

\item {\em Multimodal continuous distributions:} 
\begin{enumerate}
\item $\psi(\theta,\phi) = \frac{1}{\pi}\bet(\frac{\theta}{\pi} \mid 4,4) \times \unif(0,2\pi)$;
\item $\psi(\theta,\phi) = \unif(0,\pi) \times \frac{1}{2\pi}\bet(\frac{\phi}{2\pi} \mid 4,4)$.
\end{enumerate}
\end{enumerate}

We use $\kappa = 10$ for all cases. Simulation results are given in Table~\ref{tab:vmf}. In terms of the mixing distribution metric \eqref{eq:mixing_calculation}, performance of PR is better than maximum likelihood in the continuous cases and vice-versa in the discrete case. For the mixture density estimation comparisons, performance of maximum likelihood is comparable to that of PR. The EM algorithm in the case of von Mises--Fisher kernels is less sensitive to initial values than for Schladitz kernels, as can be seen through the results in Table~\ref{tab:vmf}. However, in terms of the continuous cases, performance of PR is still better as it actually gives a continuous estimate of a mixing density.  

\begin{table}[t]
    \centering
    \begin{tabular}{ccccc}
    \hline
    Mixing density & $\text{KL}_{\text{PR}}$ & $\text{KL}_{\text{ML}}$ & $d(\Psi^\dagger, \widehat\Psi^\dagger)_{\text{PR}}$ & $d(\Psi^\dagger, \widehat\Psi^\dagger)_{\text{ML}}$\\
     %& mean(se) & mean(se) & mean(se) & mean(se)\\
    \hline
    1 & 0.004 (0.000) & 0.002 (0.000) & 1.039 (0.007) & 0.304 (0.016)\\
    %\hline
    2 & 0.002 (0.000) & 0.005 (0.000) & 0.268 (0.005) & 1.535 (0.011)\\
    %\hline
    3 & 0.004 (0.000) & 0.003 (0.000) & 0.466 (0.011) & 1.760 (0.011)\\
    %\hline
    4 & 0.003 (0.001) & 0.005 (0.000) & 0.570 (0.008) & 1.901 (0.009)\\
    %\hline
    5(a) & 0.005 (0.000) & 0.011 (0.000) & 0.443 (0.006) & 1.927 (0.020)\\
    %\hline
    5(b) & 0.012 (0.000) & 0.019 (0.001) & 0.393 (0.008) & 1.915 (0.013)\\
    \hline
    \end{tabular}
    \caption{Simulation results using the von Mises--Fisher distribution as the kernel density.  Entries correspond to the mean of each feature, with standard errors in parentheses. KL denotes the Kullback--Leibler divergence of the mixture density estimate from the truth.}
    \label{tab:vmf}
\end{table}

\section{Extensions and applications}
\label{S:extensions}

%Available literature on mixtures on a sphere have applications in deconvoluting a discrete mixture or clustering directional data. We estimate a mixing density on a sphere as described in section \ref{PR}, which in some sense is ``more'' information than what is desired in such problems. In this section we show a couple of illustrations where the estimate obtained can be easily used to answer questions posed for mixtures on a sphere.

\subsection{Goodness-of-fit testing}
\label{goodnessoffit}

Goodness-of-fit testing is a classical inference problem wherein two candidate models are compared in terms of how well they fit a particular data set.  \citet{mardia2009} describes some of the standard nonparametric tests in the context of directional data.  An important special case of the general goodness-of-fit testing is comparing a single parametric model with a mixture thereof.  This is a highly non-trivial problem because model singularities significantly affect the standard distribution theory. So, despite its long history, with lots of proposed solutions (e.g. \citet{chen2009}; \cite{mclachlan1987}), the mixture model testing problem arguably remains unsolved.  Below we generally follow a Bayesian formulation \citep[e.g.,][]{berger2001, mcvinish2009, tokdar2010, tokdar2019}, but we make novel use of PR and the PR marginal likelihood to avoid directly specifying a full nonparametric model under the alternative.  

%A commonly asked problem is whether a mixture is required to model the data and if yes, how should the number of components be determined. \cite{mardia2009} summarize the non-parametric goodness of fit testing methods available for spherical data. However, if a mixture is required, the computations involved are non-trivial and dependent on the structure of a kernel distribution. We propose using a Bayes factor approach where a single component model can be tested versus the mixture of the same with unknown number of components.

Suppose $Y^{n} = (Y_1, \dots, Y_n)$ are iid observations from distribution with density $f$.  The goodness-of-fit test aims to decide between the two hypotheses 
\[ H_0: f \in \FF_0 \quad \text{and} \quad H_1: f \in \FF_1, \]
where $\FF_0$ and $\FF_1$ are two candidate models for $f$.  In our context, we are thinking of $\FF_0$ as a parametric family of distributions, like that described in Section~\ref{SS:prml}, i.e., 
\[ \FF_0 = \{k_\lambda(\cdot \mid \mu): \lambda \in \Lambda, \, \mu \in \SS\}, \]
where $k_\lambda(\cdot \mid \mu)$ is, say, a Schladitz or von Mises--Fisher kernel, $\mu$ is the directional parameter, and $\lambda$ is $\beta$ or $\kappa$, respectively.  Similarly, $\FF_1$ consists of kernels in $\FF_0$ mixed over the directional parameter, i.e., 
\[ \FF_1 = \{f_{\lambda,\psi}: \lambda \in \Lambda, \, \text{$\Psi$ is a mixing distribution on $\SS$}\}, \]
where $f_{\lambda,\Psi}(y) = \int k_\lambda(y \mid x) \, \Psi(dx)$ like in \eqref{eq:mixture2}.  We will also assume that $\Psi$ has a density $\psi$ with respect to the surface area measure $\sigma$ on the sphere.  

In general, the Bayesian approach to this problem proceeds by introducing prior distributions for the density $f$ under both $H_0$ and $H_1$, calculate the respective marginal likelihoods, namely, $m_0(Y^n)$ and $m_1(Y^n)$, and write the Bayes factor as 
\[ B(y^n) = m_0(y^n) / m_1(y^n). \]
Then the data shows a lack of support for $H_0$ if and only if the Bayes factor is relatively small.  For our specific setting, we consider the following prior formulations.
\begin{itemize}
\item Under $H_0$, the density is finite dimensional, determined by $\mu$ and $\lambda$, which we will model as independent {\em a priori}.  For the directional parameter $\mu$ in $\SS$, we take a uniform prior; for the parameter $\lambda$, we take a prior density $g$ (with respect to Lebesgue measure) on $\Lambda$.  Then the marginal likelihood is defined by the integral 
\[ m_0(y^n) = \int_\Lambda L_n^{\text{\sc m},0}(\lambda) \, g(\lambda) \, d\lambda, \]
where $L_n^{\text{\sc m},0}(\lambda) = (4\pi)^{-1} \int_{\SS} \prod_{i=1}^n k_\lambda(y_i \mid x) \, \sigma(dx)$ is the marginal likelihood for $\lambda$.  
\item Under $H_1$, the density is infinite dimensional because it is determined by the mixing density $\psi$ and the non-mixing parameter $\lambda$; again, we model these two as independent {\em a priori}.  For $\lambda$, we use the same prior $g$ as above.  It would be natural to model the mixing density $\psi$ or, more precisely, the mixing distribution $\Psi$ with a Dirichlet prior distribution, say, $\Pi$, with the uniform distribution on $\SS$ as the base measure---to align with the prior under $H_0$.  Then the marginal likelihood can be obtained by integrating the likelihood function $L_n(\lambda, \Psi)$ in \eqref{eq:lik} with respect to the prior distribution for $(\lambda, \Psi)$.  Here, instead, we make use of the connection between Dirichlet process mixtures and the PR algorithm---discussed in Section~\ref{SS:prml}---to approximate this marginal likelihood.  Indeed, we propose to use 
\[ m_1(y^n) = \int_\Lambda L_n^{\text{\sc m},1}(\lambda) \, g(\lambda) \, d\lambda, \]
where $L_n^{\text{\sc m},1}(\lambda)$ is the PR-based approximation of the marginal likelihood for $\lambda$.  Inside this approximation, to align with $H_0$, the PR algorithm is run with initial guess $\psi_0$ equal to the uniform prior density on $\SS$.  
\end{itemize} 

To simplify the computations further, we propose to use a Laplace approximation to evaluate both the numerator and denominator integrals in the Bayes factor.  That is, 
\[ B(y^n) \approx \frac{g(\hat\lambda_n^0)}{g(\hat\lambda_n^1)} \, \frac{L_n^{\text{\sc m},0}(\hat\lambda_n^0)}{L_n^{\text{\sc m},1}(\hat\lambda_n^1)} \, \left\{ \frac{(-\log L_n^{\text{\sc m},1})''(\hat\lambda_n^1)}{(-\log L_n^{\text{\sc m},0})''(\hat\lambda_n^0)} \right\}^{1/2}, \]
where $\hat\lambda_n^0$ and $\hat\lambda_n^1$ are the maximum marginal likelihood estimates under $H_0$ and $H_1$, respectively, and the double-prime notation denotes a second derivative. For the actual computation of the above, we use the {\tt optim} function in R to find the maximum likelihood estimates and a numerical approximation by Richardson extrapolation for calculating the second derivative.

% \begin{figure}
%     \centering
%     \includegraphics[width=10cm]{Rock_data.pdf}
%     \caption{An equal area Lambert projection of the orientations of joint planes in spherical coordinates. The orientation of the figure is with the Z-axis pointing towards us, perpendicular to the X-Y plane with the positive X-axis pointing to the right.}
%     \label{fig:Rock_data}
% \end{figure}

As an illustration of this approach, consider the example presented in \citet{fisher1993} on directional data from Triassic sandstone
at Wanganderry Lookout, New South Wales. The data set consists of $n=221$ observations which are orientations of joint planes. These measurements are converted to spherical coordinates according to the ISO convention and represented in Figure~\ref{fig:Rock_data}. Most of the observations are bunched in clusters spread around the equator indicating either a girdle or a bipolar mixture model. Below we consider two formulations of the goodness-of-fit problem, one with a Schladitz kernel and the other with a von Mises--Fisher kernel.  

\begin{figure}
    \centering
    \includegraphics[width=10cm]{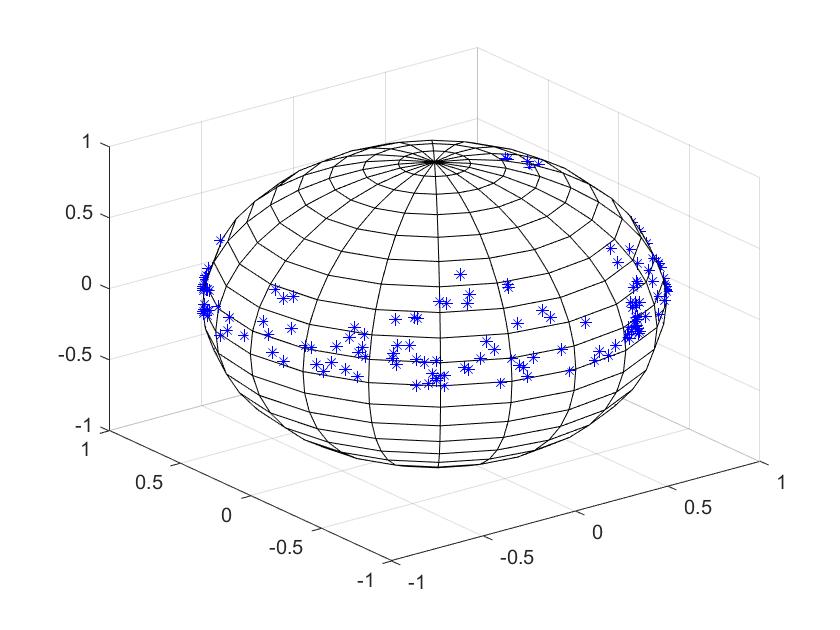}
    \caption{Orientations of the joint planes indicated by asterisks on the sphere}
    \label{fig:Rock_data}
\end{figure}

\begin{figure}
    \centering
    \subfloat[]{{\includegraphics[width=7.5cm]{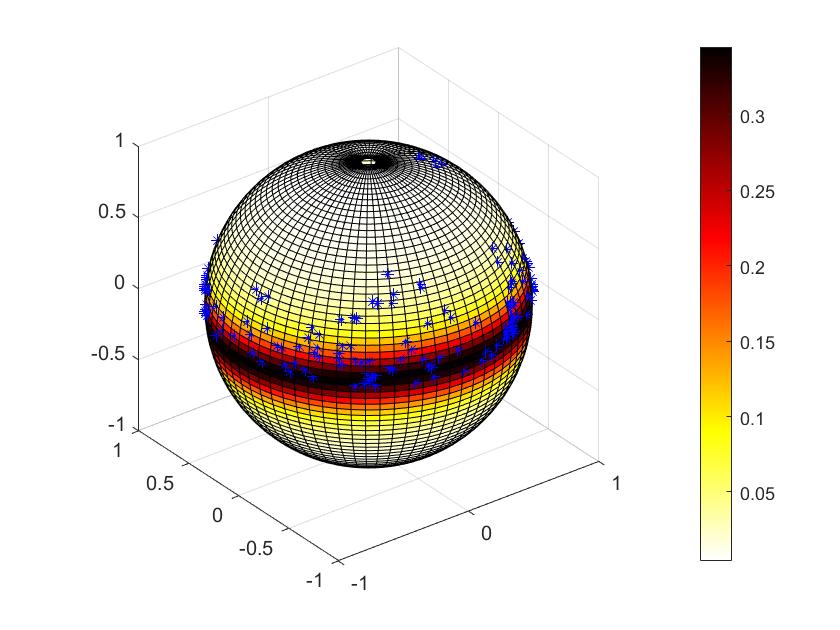}}}%
    \qquad
    \subfloat[]{{\includegraphics[width=7.5cm]{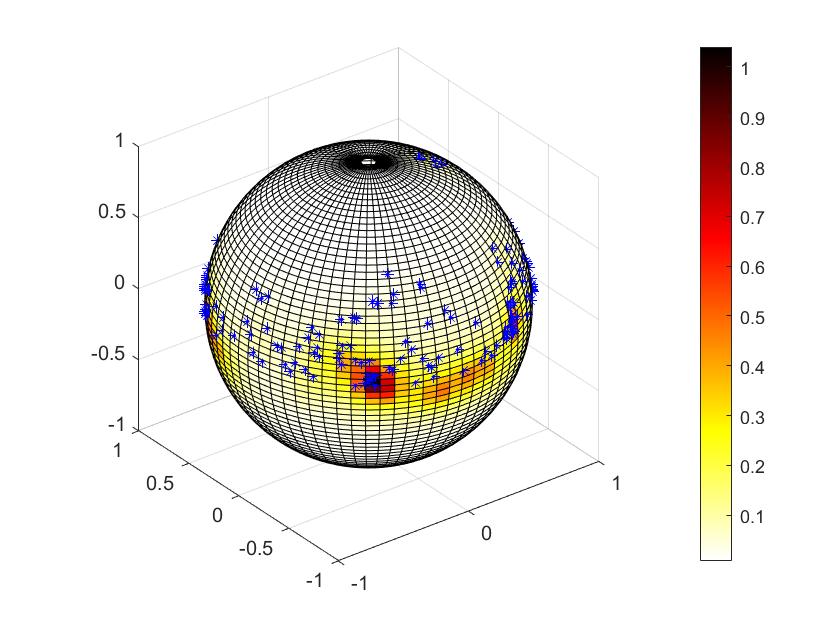}}}%
    \caption{Panels (a) and (b) represent the contour plots of the fitted Schladitz models under $H_0$ and $H_1$ respectively. The blue asterisks represent the observations.}  %
    \label{fig:schladitz_goodnessoffit}%
    \end{figure}

\begin{enumerate}
\item {\em Schladitz versus Schladitz mixtures}. For a goodness of fit test of a single Schladitz distribution versus a mixture of the same, we take the kernel $k_{\lambda}(y \mid x)$ as the Schladitz probability density, where $\lambda$ = $\beta$. The prior density for $\beta$, $g(\beta)$ is taken to be a gamma distribution with shape parameter 2 and scale parameter 0.5.  The Bayes factor is obtained using Laplace approximation as explained above. Our calculations resulted in a large Bayes factor ($> 10^9$) indicating that the simpler model under $H_0$ is a better fit and the complex mixture is unnecessary in this case. Observe Figure~\ref{fig:schladitz_goodnessoffit} which represents the fitted models under $H_0$ and $H_1$ respectively. The fitted distribution in (b) has peaks in the density where no observations are observed, which is contradictory to the data driven philosophy of PR. Recall that the Schladitz distribution is antipodally symmetric implying that for every ``peak'' in the upper hemisphere it has to have the same peak in the lower hemisphere. Since the data itself does not come from an antipodal model we see this anomalous behavior. Clearly the simpler model (under $H_0$) is a better fit to the data.

% {\color{red} Say just a few things about the model setup, i.e., that the general quantity $\lambda$ is the Schladitz shape parameter $\beta$ and than you take a specific gamma prior.  Then discuss the testing results and the plots...}. 

\begin{figure}
    \centering
    \subfloat[]{{\includegraphics[width=7.5cm]{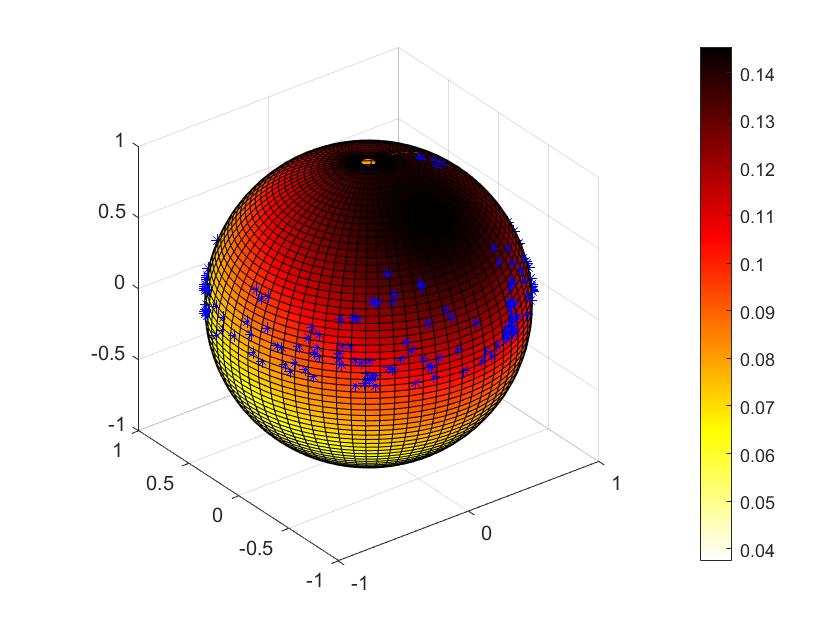}}}%
    \qquad
    \subfloat[]{{\includegraphics[width=7.5cm]{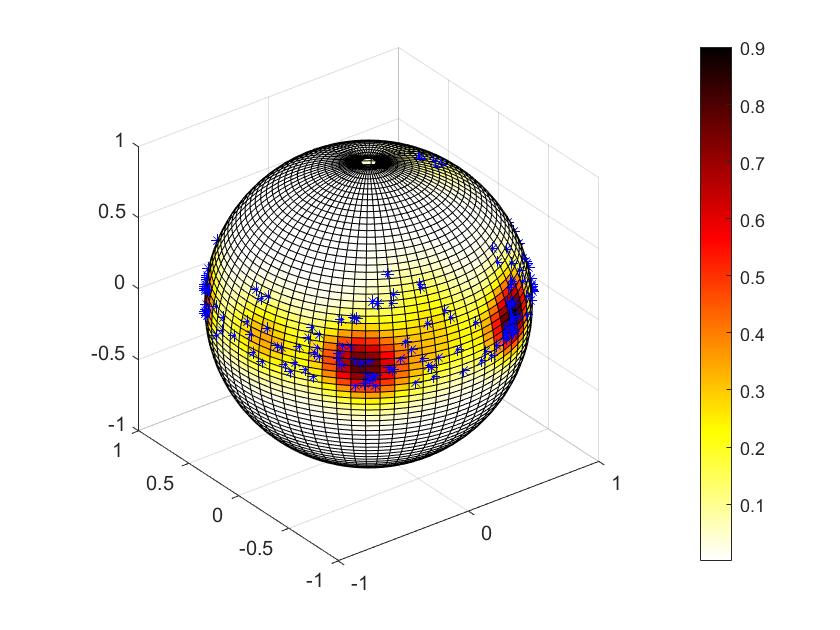}}}%
    \caption{Panels (a) and (b) represent the contour plots of the fitted von Mises--Fisher mixtures under $H_0$ and $H_1$ respectively. The blue asterisks represent the observations.}  %
    \label{fig:vmf_goodnessoffit}%
\end{figure}

\item {\em von Mises--Fisher versus von Mises--Fisher mixtures}.  For a comparison of a single von Mises--Fisher distribution versus a mixture of the same, $k_{\lambda}(y \mid x)$ is the vMF probability density, with $\lambda$ = $\kappa$. $g(\kappa)$ is taken to be the same gamma distribution as in the Schladitz case. The Bayes factor in this case is very small ($ < 10^{-10}$) indicating that the alternative model, i.e., the mixture, is preferred in this scenario. Also, Figure~\ref{fig:vmf_goodnessoffit} shows that the mixture is better able to represent clusters of data than the single distribution. Hence, the mixture is a better fit in this case. 

% {\color{red} Say just a few things about the model setup, i.e., that the general quantity $\lambda$ is the Schladitz shape parameter $\beta$ % and than you take a specific gamma prior.  Then discuss the testing results and the plots...Resolved?}. 
\end{enumerate}

\subsection{Identifying clusters from point-cloud data}

A Kinect sensor can be used to obtain depth images of a scene. Roughly, this stores the distance of a point to the sensor, in addition to the colored image. From this {\em point-cloud} data, surface normals can be extracted for further exploration. Naturally, given the various depths in the image, these normals tend to cluster on a unit sphere based on similar depths. These clusters can be viewed as missing labels representing similar orientations of planes. Identification of these clusters helps to sectionalize the image thus helping to recreate it. Our goal here is to fit a mixture model on these surface normals with the mixing density representing the clustering distribution. Then the PR estimate of the mixing density can be used for image reconstruction.  This illustration is based on point cloud data obtained from \citet{straub2015}. The data consists of smoothed surface normals extracted from depth images of the NYU V2 dataset. 
% They have used a Dirichlet Process Mixture Model, implementing clustering by sampling cluster labels from the posterior distribution.

% \begin{figure}
%     \centering
%     \includegraphics[width=10cm]{schladitz_justification.pdf}
%     \caption{An equal area Lambert projection of a random subset of point cloud data. The orientation of the figure is with the Z-axis pointing towards us, perpendicular to the X-Y plane with the positive X-axis pointing to the right. The black squares indicate points in the upper hemisphere and the red triangles represent points in the lower hemisphere which is rotated towards the positive Z axis. }
%     \label{fig:schladitz_justification}
% \end{figure}

\begin{figure}
    \centering
    \includegraphics[width=10cm]{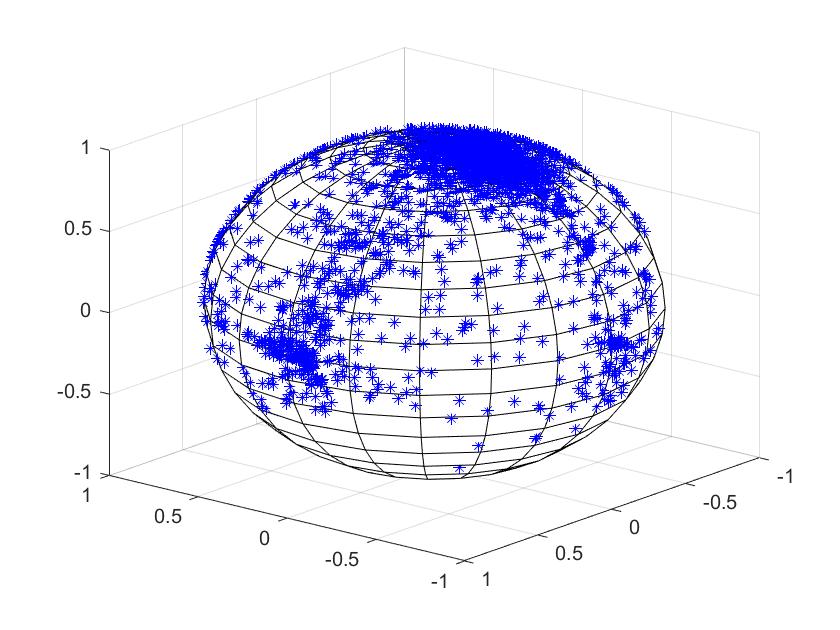}
    \caption{A subset of the NYU data set.}
    \label{fig:ny_subset}
\end{figure}

Let the smoothed surface normals be denoted by $Y_1, \ldots, Y_n$. \cite{straub2015} model the probability density of the missing labels (orientation of the planes) as a discrete distribution with infinite support points. Let $\psi (x)$ denote this density, which plays the role of the mixing density as in  equation \ref{eq:mixture}. Recall that we do not need to make any assumptions regarding this distribution. To use PR for estimation of the mixing distribution, we need to specify an appropriate kernel. Since the data set in question is huge (~233,000 raw data points, removing the missing observations), plotting the entire data set does not give us a clear idea of its distribution, hence we look at a random subset of these surface normals as in Figure~\ref{fig:ny_subset}. The clusters formed in the plot do not suggest an antipodal model due to an absence of equivalent clusters in the upper and lower hemisphere. Hence, we propose the use of a von Mises--Fisher kernel.

Estimation of the latent mixing density first involves estimating the structural parameter $\kappa$ associated with the kernel distribution. Using the entire data set for estimating $\kappa$ is computationally inefficient even with fast computation of PR. Hence we use a random subset of 500 observations to estimate $\kappa$ using the maximization of marginal likelihood method as proposed in Section~\ref{SS:prml}. Using this estimate, the mixing density $\psi$ is estimated using PR as in Section~\ref{SS:prcomp}. Figure~\ref{fig:latentdensity} shows a contour plot of the mixing density estimate on the $(\theta,\phi)$ grid. As we can see from the figure, the distribution has around 4--5 modal locations representing the depths in the image. The idea is to assign each $Y_i$ with a cluster label based on this mixing distribution. For this we maximize the posterior probability for each $Y_i$ with respect to these clusters, hence obtaining an assignment. To represent this labelling, each cluster label is allotted a color and an image is constructed based on these colors at each respective $Y_i$ as in the original image. The reconstruction is represented in Figure~\ref{fig:PRpointcloud}, where we can see that the estimated clustering has been able to identify the depths in the original image.

\begin{figure}[t]
    \centering
    \includegraphics[width=10cm]{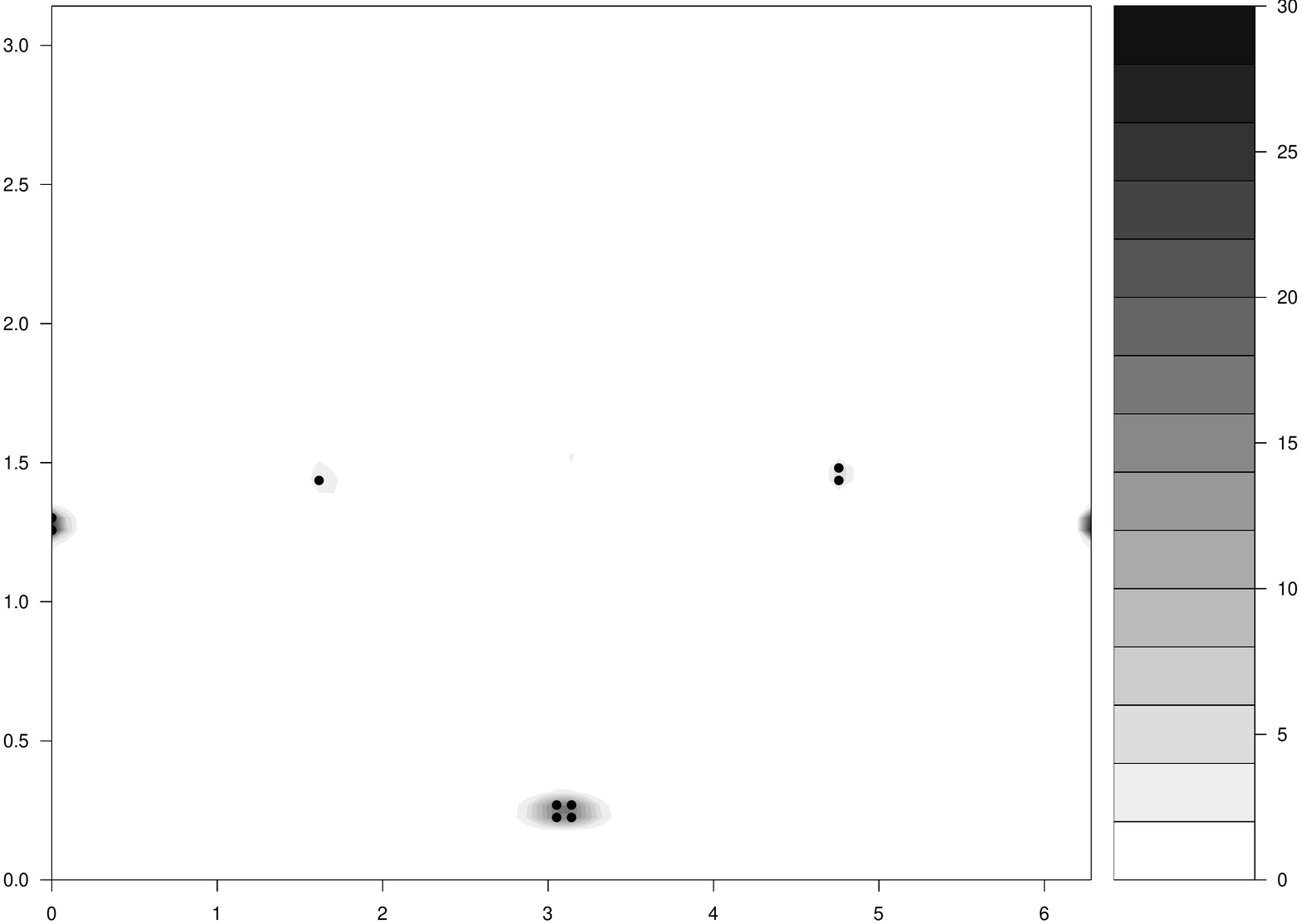}
    \caption{A contour plot of the estimated mixing density on the $(\theta,\phi)$ grid, where the X-axis represents the azimuthal angle $\phi$ and the Y-axis represents the polar angle $\theta$. Black spots indicate the modes used for clustering on the sphere}
    \label{fig:latentdensity}
\end{figure}

\begin{figure}[t]
    \centering
    \subfloat[Original image]{{\includegraphics[width=7cm, height=7cm]{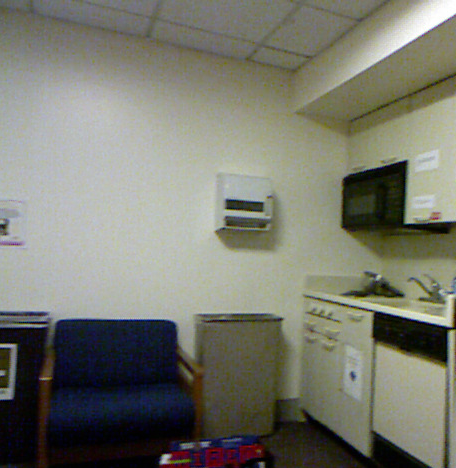}}}
    \qquad
    \subfloat[Clustering]{{\includegraphics[width=7cm, height=7cm]{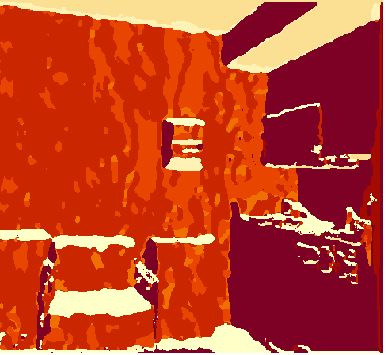}}}
    \caption{Original image and the reconstructed image based on the PR-based clustering.}%
    \label{fig:PRpointcloud}%
\end{figure}

\section{Summary and future work}
\label{S:discuss}

In this paper, we have proposed using the PR algorithm to estimate a smooth mixing density when modeling directional data by a mixture. Practical advantages of the PR algorithm include its speed and its flexibility, i.e., the algorithm does not require tailoring to a particular form of kernel as does the EM algorithm.  The PR estimator is consistent under suitable conditions, and we show that these assumptions hold for both Schladitz and von Mises--Fisher kernels.  Further, we show that any additional structural parameters associated with the kernel can be estimated via maximizing the PR-based marginal likelihood. We show the efficiency of this method using both vectorial and axial models.  Estimates of the mixing density can be used to answer various questions posed for mixtures on a sphere.  We considered two such problems, namely, goodness-of-fit testing and clustering, and proposed novel PR-based methods to address them.  

%We propose 2 such illustrations. Using a real data set of orientations of joint planes we conduct the goodness of fit method that compares a single component model to a mixture. This method is computationally efficient through the use of the PR algorithm as an approximation for the Dirichlet Process mixture marginal likelihood. The other application is based on point cloud data. We have showed that an image can be {\color{red}sectionalized} by clustering the surface normals based on modes of the estimated mixing density.

There are several open questions concerning both mixtures of spherical distributions and the PR algorithm.  First, to our knowledge, identifiability of general (non-finite) mixtures of spherical kernels has not been established.  That finite versions of these mixtures have, at least in some cases, been shown to be identifiable \citep[e.g.,][]{kent1983} gives us every reason to expect that identifiability holds for general mixtures too, but the details still need to be worked out.  Second, general asymptotic properties of the maximum PR marginal likelihood estimator described in Section~\ref{SS:prml} are lacking.  We fully expect consistency, asymptotic normality, etc., to hold under suitable conditions but, again, the details remain to be worked out.

%\newpage

\section*{Acknowledgments}

This work is partially supported by the U.S.~National Science Foundation, grant DMS--1737929.  The authors thank Prof.~Dr.~Claudia Redenbach for sharing her code for fitting mixtures of Schladitz distributions.  

\appendix

\section{Appendix}
\label{appendix}

Theorem~\ref{thm:pr} states the asymptotic convergence properties of the PR estimates of the mixing and mixture distributions.  Part~1 of the theorem concerns the convergence properties of the PR estimate of the mixture density, and below we state the precise conditions required for that theorem and show that these are satisfied for the kinds of mixture models considered here in this paper.  Part~2 of the theorem concerns the properties require an extra identifiability assumption, which we will not discuss here; see Section~\ref{S:discuss}.  

Recall that $f^\star$ is the true density for the iid data, which may or may not be of the specified mixture model form in \eqref{eq:mixture}.  There are several general conditions---mostly about $f^\star$ and the kernel $k$---stated in \citet{tmg} and/or \citet{martintokdar2009}, some of which are more-or-less automatic in the present spherical mixtures setting.  First, when the mixture model is misspecified, Condition~A1 in \citet{martintokdar2009} requires that the set of all mixing distributions---absolutely continuous with respect to the specified dominating measure---be (relatively) compact.  When the support of the mixing distribution is compact, as it is in our present setting, then this follows from Prokhorov's theorem.  Condition~A2 requires that the kernel $x \mapsto k(y \mid x)$ be bounded and continuous for each $y$, which is easy to check for Schladitz and von Mises--Fisher kernels.  Except for the decay condition on the user-specified weight sequence, as discussed in Section~\ref{SS:algorithm}, there is one more non-trivial condition required for Part~1 of the theorem, namely, the following (integrability) Condition~A4 in \citet{martintokdar2009}:
\begin{equation}
\label{eq:integrable}
\sup_{x_1, x_2 \in \SS} \int \Bigl\{ \frac{k(y \mid x_1)}{k(y \mid x_2)} \Bigr\}^2 \, f^\star(y) \, \sigma(dy) < \infty. 
\end{equation}
We will check this below, in Propositions~1--2, for the Schladitz and von Mises--Fisher kernels, respectively.  For that, we should mention that there is a Condition~A6 in \citet{martintokdar2009} that is generally required for convergence of the PR estimate of the mixing distribution.  This boils down to a sort of tightness on the class of kernels, which follows immediately by compactness of the sphere.  

\begin{prop}
Condition \eqref{eq:integrable} is satisfied by the the Schladitz kernel for every $f^\star$.  
\end{prop} 

\begin{proof}
When $\beta > 0$ is fixed, the Schladitz kernel can be written as 
\[ k(y \mid x) \propto |\Sigma_x|^{-1/2} (y^\top \Sigma_x^{-1} y)^{-3/2}, \]
where $\Sigma_x = Q_x^\top D_\beta Q_x$, with  $D_\beta=\text{diag}(1,1,\beta^{-2})$ and $Q_x$ is the rotation matrix that maps the unit vector $(0,0,1)^\top$ onto $x$.  Plugging the kernel density, at two generic unit vectors $x_1$ and $x_2$, into the left-hand side of \eqref{eq:integrable}, we get 
\begin{align*}
\int \Bigl\{ \frac{k(y \mid x_1)}{k(y \mid x_2)} \Bigr\}^2 \, f^\star(y) \, \sigma(dy) & = \int \frac{|\Sigma_{x_2}|}{|\Sigma_{x_1}|} \Bigl( \frac{y^\top \Sigma_{x_2}^{-1} y}{y^\top \Sigma_{x_1}^{-1} y} \Bigr)^3 \, f^\star(y) \, \sigma(dy) \\
& = \frac{|\Sigma_{x_2}|}{|\Sigma_{x_1}|} \int \Bigl( \frac{y^\top \Sigma_{x_2}^{-1} y}{y^\top \Sigma_{x_1}^{-1} y} \Bigr)^3 \, f^\star(y) \, \sigma(dy).
\end{align*}
Since $Q_x$ is a rotation matrix for all $x$, and since $\beta$ is fixed, the ratio in front of the integral in the last expression is 1.  Note also that
\[ e_{\text{min}}(D_\beta) \|Q_x y\|^2 \leq (Q_x y)^\top D_\beta (Q_x y) \leq e_{\text{max}}(D_\beta) \|Q_x y\|^2, \]
where $e_{\text{min}}(D_\beta)$ and $e_{\text{max}}(D_\beta)$ are the minimum and maximum eigenvalues of $D_\beta$, which are fixed numbers, and that $\|Q_xy\|^2 = \|y\|^2$.  Then we immediately see that the ratio of quadratic forms is uniformly bounded in both $(x_1,x_2)$ and $y$, i.e., 
\[ \Bigl( \frac{y^\top \Sigma_{x_2}^{-1} y}{y^\top \Sigma_{x_1}^{-1} y} \Bigr)^3 \leq \text{constant}. \]
Therefore, the expectation with respect to $f^\star$ is also bounded, hence \eqref{eq:integrable} holds. 
\end{proof}

\begin{prop}
If $f^\star$ admits a continuous moment generating function 
\[ t \mapsto \int \exp(t^\top y) \, f^\star(y) \, \sigma(dy), \]
then condition \eqref{eq:integrable} is satisfied by the von Mises--Fisher kernel.  
\end{prop} 

\begin{proof}
Plugging in the von Mises--Fisher kernel and simplifying gives 
\begin{align*}
\int \Bigl\{ \frac{k(y \mid x_1)}{k(y \mid x_2)} \Bigr\}^2 \, f^\star(y) \, \sigma(dy) & = \int \frac{C_3(\kappa) \exp(2\kappa x_1^\top y)}{C_3(\kappa) \exp(2\kappa x_2^\top y)} \, f^\star(y) \, \sigma(dy) \\
& = \int \exp\{2 \kappa (x_1 - x_2)^\top y\} \, f^\star(y) \, \sigma(dy). 
\end{align*}
The last expression is just the moment generating function of $f^\star$ evaluated at $t=2\kappa(x_1-x_2)$.  Since $\kappa$ is fixed, and $x_1$ and $x_2$ are on the sphere, this argument $t$ ranges over a compact set.  Since the moment generating function is assumed to be continuous, it must attain its (finite) maximum on that compact, hence \eqref{eq:integrable}.
\end{proof}

\bibliographystyle{apalike}
\bibliography{refs}

\end{document}